\RequirePackage{fix-cm}
\RequirePackage{amsmath}
\documentclass[smallcondensed]{svjour3}
\smartqed
\usepackage{graphicx}
\usepackage{times}

\usepackage{amsfonts}
\usepackage{color}
\usepackage{tikz}
\usetikzlibrary{shapes,arrows}
\usetikzlibrary{arrows,fit,shapes,automata}
\usetikzlibrary{positioning,fit,calc,shapes}
\usepackage{ltl} 
\usepackage{xspace} 
\usepackage{booktabs,tabularx}
\usepackage{multirow,arydshln}
\usepackage{paralist}
\usepackage{graphicx,graphics,epsf,caption,subcaption}
\usepackage{wrapfig}

\newcommand{\true}{\mathit{true}}
\newcommand{\false}{\mathit{false}}
\newcommand{\trueval}{\mathsf{true}}
\newcommand{\falseval}{\mathsf{false}}

\newcommand{\nats}{\mathbb{N}}
\newcommand{\natsbot}{\nats \cup \{\bot\}}

\newcommand\defeq{\ensuremath{\mathrel{\raisebox{-.3ex}{$\stackrel{\text{\tiny def}}=$}}}\xspace}

\newcommand{\bigo}[1]{\mathcal{O}(#1)}

\newcommand{\subf}[1]{\mathsf{subf}(#1)}

\newcommand{\trans}{\mathcal{T}}
\newcommand{\Traces}{\mathsf{Traces}}
\newcommand{\val}[2]{\mathit{val}(#1,#2)}

\newcommand{\autA}{\mathcal{A}}
\newcommand{\autB}{\mathcal{B}}
\newcommand{\autN}{\mathcal{N}}
\newcommand{\lang}[1]{\mathcal{L}(#1)}
\newcommand{\relaxfg}{\mathit{Relax}_{\tiny\fg}}

\newcommand{\rej}{\mathsf{rej}}
\newcommand{\Rej}{\mathit{Rej}}

\newcommand{\props}{\mathcal{P}}
\newcommand{\alphabet}{\Sigma}
\newcommand{\ialphabet}{{2^\mathcal{I}}}
\newcommand{\oalphabet}{{2^\mathcal{O}}}
\newcommand{\inpv}{\mathcal{I}}
\newcommand{\outv}{\mathcal{O}}
\newcommand{\inpval}{{\sigma_{I}}}
\newcommand{\outval}{{\sigma_{O}}}

\newcommand{\spec}{{\varphi}}
\newcommand{\softSpec}{{\varphi}}
\newcommand{\fg}{\LTLfinally\LTLglobally}
\newcommand{\gf}{\LTLglobally\LTLfinally}
\newcommand{\gphi}{{\LTLglobally\varphi}}
\newcommand{\gpsi}{{\LTLglobally\psi}}
\newcommand{\gphij}{{\LTLglobally\varphi_j}}
\newcommand{\fgphi}{{\LTLfinally\LTLglobally\varphi}}
\newcommand{\fgpsi}{{\LTLfinally\LTLglobally\psi}}
\newcommand{\gfphi}{{\LTLglobally\LTLfinally\varphi}}
\newcommand{\gfpsi}{{\LTLglobally\LTLfinally\psi}}
\newcommand{\Relax}{\mathit{Relax}}

\newcommand{\anot}{\lambda}
\newcommand{\anotfg}{\pi}
\newcommand{\anotb}{\lambda^{\mathbb{B}}}

\newcommand{\anotfgbj}{\pi^{\mathbb{B},j}}
\newcommand{\anotbj}{\lambda^{\mathbb{B},j}}
\newcommand{\anotbjk}{\lambda^{\mathbb{B},j,k}}
\newcommand{\anotbjl}{\lambda^{\mathbb{B},j,l}}
\newcommand{\anotn}{\lambda^{\mathbb{N}}}

\newcommand{\anotfgnj}{\pi^{\mathbb{N},j}}
\newcommand{\anotnj}{\lambda^{\mathbb{N},j}}
\newcommand{\fgvalid}{$\LTLfinally\LTLglobally$--valid}
\newcommand{\succa}{\mathsf{succ}}
\newcommand{\soft}{\mathit{Soft}}

\newcommand{\supplies}{\mathit{P}}
\newcommand{\loads}{\mathit{L}}
\newcommand{\slp}{s_{l \rightarrow p}}
\newcommand{\caps}{E^+}
\newcommand{\cons}{\mathit{Consumers}}
\newcommand{\supl}{\mathit{Suppliers}}
\newcommand{\office}{\mathit{office}}
\newcommand{\occupied}{\mathit{occupied}}
\newcommand{\passage}{\mathit{passage}}

\newcommand{\library}{\mathit{library}}
\newcommand{\ent}{\mathit{entrance}}
\newcommand{\corr}{\mathit{corridor}}
\newcommand{\exh}{\mathit{exhibition}}

\makeatletter
\renewcommand{\fnum@figure}{Figure \thefigure}
\makeatother

\journalname{Acta Informatica}

\begin{document}

\title{Reactive Synthesis with Maximum Realizability of \\ Linear Temporal Logic Specifications 
	\thanks{Part of the results in this paper were
		presented at the Sixteenth International Symposium on Automated Technology for Verification and Analysis, Los Angeles, California, USA, October 2018 \cite{dimitrova2018maximum}.
}
}

\author{Rayna Dimitrova$^{\ast}$ \and
        Mahsa Ghasemi$^{\ast}$ \and \\
        Ufuk Topcu
        \thanks{$\ast$These authors contributed equally to the manuscript.}
}

\institute{Rayna Dimitrova \at
              University of Leicester, Leicester, UK \\
              \email{rd307@leicester.ac.uk}
           \and
           Mahsa Ghasemi \at
              University of Texas at Austin, Austin, Texas, USA \\
              \email{mahsa.ghasemi@utexas.edu}
           \and
           Ufuk Topcu \at
              University of Texas at Austin, Austin, Texas, USA \\
			  \email{utopcu@utexas.edu }           
}

\date{}

\maketitle

\begin{abstract}
A challenging problem for autonomous systems is to synthesize a reactive controller that conforms to a set of given correctness properties. Linear temporal logic (LTL) provides a formal language to specify the desired behavioral properties of systems. In applications in which the specifications originate from various aspects of the system design, or consist of a large set of formulas, the overall system specification may be unrealizable. Driven by this fact, we develop an optimization variant of synthesis from LTL formulas, where the goal is to design a controller that satisfies a set of hard specifications and minimally violates a set of soft specifications. To that end, we introduce a value function that, by exploiting the LTL semantics, quantifies the level of violation of properties. 
Inspired by the idea of bounded synthesis, we fix a bound on the implementation size and search for an implementation that is optimal with respect to the said value function. We propose a novel maximum satisfiability encoding of the search for an optimal implementation (within the given bound on the implementation size). We iteratively increase the bound on the implementation size until a termination criterion, such as a threshold over the value function, is met.

\keywords{Maximum realizability \and Linear temporal logic \and Bounded synthesis \and Maximum satisfiability}

\end{abstract}

\section{Introduction}\label{sec:intro}
In an ideal world, a user may specify a set of high-level behavioral characteristics for an autonomous system, and a controller (i.e., implementation) can be synthesized to comply with these specifications. Such automatic synthesis has been a topic for various studies in the domain of formal methods where the goal is to design a hardware or a software system that satisfies a set of formally defined properties. These properties can be formulated in an appropriate language, such as linear temporal logic (LTL)~\cite{pnueli1977temporal}. In conventional synthesis, either an implementation is constructed for a given specification, or the specification is identified as unrealizable. Nevertheless, especially in large systems, specifications may arise from different design perspectives, and if they consist of a large number of individual requirements, it is easy to encounter specifications that are unrealizable. In other scenarios, the user may have several alternative requirements in mind, potentially with some preferences, and want to know the best realizable combination of them with respect to some metric. Such cases usually lead to alternating between specification modification and synthesis procedure and hence, defeating the purpose of facilitating the design process. 

The possibility of conflict amongst the provided requirements calls for a more comprehensive synthesis procedure that, in the case of unrealizability, can generate an implementation that minimally violates the specifications. In order to define the notion of minimality, one requires a quantitative metric on the satisfaction of LTL formulas. The approach we pursue relies on multiple levels of relaxations of an LTL formula. We associate each level with a binary variable and form a value function that indicates the levels of relaxations of the formula that the implementation satisfies. The value function respects a lexicographic ordering according to preferences over relaxations. Having defined a value function, one can interpret maximum realizability of a set of LTL formulas as seeking an implementation that maximizes the corresponding value. 

In this paper, we consider settings in which the goal is to design a system that satisfies a given hard specification and maximizes the defined value function over a set of (potentially prioritized) soft specifications. We first focus on soft specifications that are safety formulas and later discuss the extension to general formulas. We quantify the compliance with a safety property according to the LTL semantics in a straightforward manner. The highest value is associated with satisfying the formula at all times, and it monotonically decreases if the formula is satisfied from some point on, then satisfied infinitely often, and lastly satisfied only for a finite number of times. Based on this ordering, we define the cumulative value of a conjunction of safety properties according to given priorities or design criteria.

The backbone of our approach toward maximum realizability is bounded synthesis, originally introduced by Schewe and Finkbeiner~\cite{ScheweF07a}. Bounded synthesis tackles the computational complexity of reactive synthesis from LTL properties by restricting the size of the search space. More specifically, it searches for a realizable implementation of the size up to a prespecified bound. If no such implementation exists, it increments the bound and repeats the search process. Each instance of bounded search for an implementation can be encoded as a SAT (or QBF, or SMT) problem~\cite{FaymonvilleFRT17}. The algorithm is complete as a theoretical bound on the maximum size of the implementation exists. 

We formulate maximum realizability as iterative MaxSAT solving~\cite{biere2009handbook}. In each iteration, we construct a MaxSAT instance that characterizes the existence of an implementation of size within the given bound that not only realizes the hard specification but is also optimal with respect to defined value function for soft specifications. We prove that, for any given finite set of soft specifications, there exists an optimal implementation with a bounded size. Consequently, the proposed algorithm that gradually increases the bound on the implementation size is complete. On the other hand, the theoretical upper bound on the minimal implementation size is generally impractically large. Therefore, in practice, we also settle for termination criteria such as a problem-specific bound on controller size, a limit on running time, or a desired threshold on the value function.

The proposed encoding of maximum realizability generates partial weighted MaxSAT instances. A partial weighted MaxSAT problem is composed of a set of hard clauses and a set of weighted soft clauses. The hard clauses capture the encodings imposed by bounded synthesis procedure for both hard and soft specifications. The soft clauses determine the level of relaxation of a soft specification that can be satisfied. We design the weights of soft clauses in a way that they correspond to the quantitative semantics of soft specifications. Therefore, adjusting the weights allows our approach to easily adapt to different design criteria.\looseness=-1

Recent advances in SAT solving along with the development of novel algorithms such as structural partitioning and search heuristics have made MaxSAT solviers a promising tool. MaxSAT formulations have been effective in solving many real-world problems including most probable explanation in Bayesian networks~\cite{park2002using}, package management~\cite{janota2012packup}, and correlation clustering and causal structure learning~\cite{berg2015applications}. While SAT solving has been an essential part of numerous formulations proposed for reactive synthesis problems, the applicability of MaxSAT solving has not yet been explored in this field. In this paper, we develop the first maximum realizability algorithm that utilizes the power of MaxSAT solvers to deal with the underlying combinatorial nature of the optimization task.

We evaluated the proposed maximum realizability procedure experimentally on reactive synthesis instances from two domains where considering combinations of hard and soft specifications is natural and often unavoidable. The first domain is robotic navigation, where due to the adversarial nature of the environment in which robots operate, safety requirements might prevent a system from achieving its goal, or a large number of tasks of different types might not necessarily be consistent when posed together.
The second domain relates to load distribution tasks in power networks. There, generators have limited capacity to power a set of vital and non-vital loads, whose total demand may exceed the capacity of the generators, thus leading to a combination of hard and soft specifications.

The rest of the manuscript is organized as follows. Section~\ref{sec:rel-work} discusses related work. Section~\ref{sec:background} recalls the necessary background on LTL synthesis, the bounded synthesis approach, and  maximum satisfiability. In Section~\ref{sec:prob}, we describe the proposed quantitative semantics for soft specifications and using this  semantics, formally state the maximum realizability problem. In Section~\ref{sec:maxsat-encoding}, we detail the proposed bounded maximum realizability algorithm along its encoding into MaxSAT instances. Section~\ref{sec:experiments} presents the experimental settings and the obtained results. Lastly, Section~\ref{sec:conclusion} states the concluding remarks and future directions.

This paper is an extension of the conference publication~\cite{dimitrova2018maximum}. It contains the complete proofs and  presents (Section~\ref{sec:generalizations}) a generalization of the maximum realizability problem to soft specifications in the full class of LTL formulas and to prioritized specifications.

\section{Related Work}\label{sec:rel-work}
Maximum realizability and several closely related problems have attracted significant attention in recent years. Tumova et al.~\cite{TumovaHKFR13} studied the problem of planning over a finite horizon with prioritized safety requirements, where the goal is to synthesize a least-violating control strategy. Kim et al.~\cite{KimFS15} studied a similar problem for the case of infinite-horizon temporal logic planning, which seeks to revise an inconsistent specification, minimizing the cost of revision with respect to costs for atomic propositions provided by the specifier. Lahijanian et al.~\cite{LahijanianAFKV15} describe a method for computing optimal plans for co-safe LTL specifications, where optimality is again with respect to the cost of violating each atomic proposition, which is provided by the user. All of these approaches are developed for the planning setting, where there is no adversarial environment, and thus they are able to reduce the problem to the computation of an optimal path in a graph. Lahijanian and Kwiatkowska~\cite{LahijanianK16} considered the case of probabilistic environments. In contrast, the proposed method seeks to maximize the satisfaction of the given specification against the worst-case behavior of the environment. Lahijanian et al.~\cite{LahijanianMFKKV16} studied the problem of partial satisfaction of guarantees in an unknown environment, where, unlike in our work, no relaxations of the soft specifications are considered, but simply the number of those that are satisfied is maximized.

The problem setting that is the closest to ours is that of Tomita et al.~\cite{Tomita2017}. There, the authors study a maximum realizability problem in which the specification is a conjunction of a \emph{must} (or \emph{hard}, in our terms) LTL specification, and a number of weighted \emph{desirable} (or \emph{soft}, in our terms) specifications of the form $\gphi$, where $\varphi$ is an arbitrary LTL formula. When $\varphi$  is not a safety property it is first strengthened to a safety formula before applying the synthesis procedure, which then weakens the result to a mean-payoff term. Thus, while Tomita et al. consider a broader class of soft specifications compared to those in this paper, when $\varphi$ is not a safety property there is no clear relationship between $\gphi$ and the resulting mean-payoff term. 
When applied to multiple soft specifications, the method by Tomita et al. combines the corresponding mean-payoff terms in a weighted sum, and synthesizes an implementation optimizing the value of this sum. Thus, without inspecting the synthesized implementation it is not possible to determine to what extent the individual desirable specifications are satisfied. In contrast, in the proposed maximum realizability procedure each satisfaction value is characterized as an LTL formula, which is useful for explainability and providing feedback to the  designer. 

To the best of our knowledge, our work is the first to employ MaxSAT in the context of reactive synthesis. MaxSAT has been used for preference-based planning~\cite{JumaHM12} and for computing optimal plans in propositional planning problems with action costs~\cite{robinson2010partial}. However, since maximum realizability is concerned with reactive systems, it requires a fundamentally different approach from planning.
 
Two other main research directions related to maximum realizability are \emph{quantitative synthesis} and \emph{specification debugging}. There are two predominant flavours of quantitative synthesis problems studied in the literature. In the first one (cf.~\cite{BloemCHJ09}), the goal is to generate an implementation that maximizes the value of a mean-payoff objective, while possibly satisfying some $\omega$-regular specification.
In the second setting (cf.~\cite{AlmagorBK16,TabuadaN16}), the system requirements are formalized in a multi-valued temporal logic. These synthesis methods~\cite{TabuadaN16,BloemCHJ09,AlmagorBK16}, however, do not directly solve the corresponding optimization problem, but instead check for the existence of an implementation whose value is in a given set. The optimization problem can then be reduced to a sequence of such queries. 

Alur et al.~\cite{AlurKW08} studied an optimal synthesis problem for an ordered sequence of prioritized $\omega$-regular properties, where the classical fixpoint-based game-solving algorithms are extended to a quantitative setting. The main difference in our work is that we allow for incomparable soft specifications each with a number of prioritized relaxations, for which the equivalent set of preference-ordered combinations would be of size exponential in the number of soft specifications. Our MaxSAT formulation avoids explicitly considering these combinations.

In specification debugging there is a lot of research dedicated to finding good explanations for the unsatisfiability or unrealizability of temporal logic specifications~\cite{CimattiRST07,Schuppan12,RamanK13}, and more generally to the analysis of specifications~\cite{CimattiRST08,EhlersR14}. Our approach to maximum realizability can prove useful for specification analysis, since instead of simply providing an optimal value, it computes an optimal relaxation of the given specification in the form of another LTL formula.

\section{Background}\label{sec:background}
We start by an overview of syntax and semantics of linear temporal logic (LTL) and  language-equivalent automata representation. Next, we define finite-state transition systems, and formally state the synthesis problem.
Then, we proceed to go over definitions of run graph and annotations to describe the bounded synthesis method and its SAT encoding.
Lastly, we provide a brief description of maximum satisfiability (MaxSAT) problem, particularly a class of that called partial weighted MaxSAT.

\subsection{Synthesis from LTL Specifications}\label{sec:def-synth}

Linear temporal logic (LTL) is a formal language for specifying behavioral characteristics of reactive systems. Formulas in LTL are constructed according to the following grammar:
\begin{equation*}
	\varphi := p \mid \true \mid \false \mid \neg \varphi \mid \varphi_1 \wedge \varphi_2 \mid \varphi_1 \vee \varphi_2 \mid \LTLnext  \varphi  \mid \varphi_1 \LTLuntil \varphi_2 \mid \varphi_1 \LTLrelease \varphi_2,
\end{equation*}
where $p \in \props$ is an atomic proposition. \emph{next} ($\LTLnext$), \emph{until} ($\LTLuntil$), and \emph{release} ($\LTLrelease$) are temporal operators. The \emph{finally} operator ($\LTLfinally$) is defined as $\LTLfinally \varphi \equiv \true \LTLuntil \varphi$ and the \emph{globally} operator ($\LTLglobally$) is defined as $\LTLglobally \varphi \equiv \false \LTLrelease \varphi$.
We denote the size of an LTL formula (that is, the number of operators in the formula) with $|\varphi|$, and the set of all its subformulas with $\subf\varphi$.
A negation normal form (NNF) of an LTL formula is a semantically-equivalent LTL formula in which negations appear only in front of atomic propositions. Without loss of generality, we consider LTL formulas in NNF.

Let $\Sigma = 2^{\props}$  denote the finite alphabet composed of all possible valuations of the propositions. A letter $\sigma \in \Sigma$ is interpreted as the valuation that assigns value $\trueval$ to all $p \in \sigma$ and $\falseval$ to all $p \in \props \setminus \sigma$. An infinite word $w \in\alphabet^\omega$ is an infinite sequence of letters. An LTL formula $\varphi$ defines a language over infinite words. A word is included in the language if it satisfies the formula, denoted by $w \models \varphi$. The full semantics of LTL can be found in~\cite{BaierKatoen08}.

A \emph{safety LTL formula} is an LTL formula such that every word not in its language has a bad prefix. Formally, $\varphi$ is a safety LTL formula if for each $w \not\models\varphi, w \in\alphabet^\omega$, there exists a bad prefix $u \in \alphabet^*$ such that $u\cdot v \not\models \varphi, \forall v\in\alphabet^\omega$. 
\emph{Syntactically safe} LTL formulas are a subclass of safety LTL formulas which  do not contain any occurrences of $\LTLuntil$ when written in NNF.

The language accepted by an LTL formula can equivalently be represented by a nondeterministic (or universal) B\"uchi (or co-B\"uchi) automaton.
A \emph{B\"uchi automaton} over a finite alphabet $\alphabet$ is a tuple $\autA = (Q,q_0,\delta,F)$, where $Q$ is a finite set of states, $q_0$ is the initial state, $\delta \subseteq Q \times \alphabet \times Q$ is the transition relation, and $F \subseteq Q$ is a subset of states. A run of $\autA$ on an infinite word $w=\sigma_0\sigma_1\ldots \in \alphabet^\omega$ is an infinite sequence $q_0,q_1,\ldots$ of states, where $q_0$ is the initial state and for every $i \geq 0$ it holds that $(q_i,\sigma_i,q_{i+1}) \in \delta$. 

A run of a B\"uchi automaton is accepting if it contains infinitely many occurrences of states in $F$. A \emph{co-B\"uchi automaton} $\autA = (Q,q_0,\delta,F)$ differs from a B\"uchi automaton in the accepting condition: a run of a co-B\"uchi automaton is accepting if it contains only \emph{finitely many} occurrences of states in $F$. For a B\"uchi automaton the states in $F$ are called \emph{accepting states}, while for a co-B\"uchi automaton they are called \emph{rejecting states}.
A \emph{nondeterministic} automaton $\autA$ accepts a word $w \in \alphabet^\omega$ if \emph{some} run of $\autA$ on $w$ is accepting.
A \emph{universal} automaton $\autA$ accepts a word $w \in \alphabet^\omega$ if \emph{every} run of $\autA$ on $w$ is accepting.

For a reactive system, the set of atomic propositions is $\props = \inpv \cup \outv$, where $\inpv$ and $\outv$ are disjoint sets, denoting \emph{input} propositions controlled by the environment and \emph{output} propositions controlled by the system, respectively.
A \emph{transition system} over a set of input propositions $\inpv$ and a set of output propositions $\outv$ is a tuple $\trans = (S,s_0,\tau)$, where $S$ is a set of states, $s_0$ is the initial state, and the transition function $\tau : S \times \ialphabet \to S \times \oalphabet$ maps a state $s$ and a valuation $\inpval \in \ialphabet$ of the input propositions to a successor state $s'$ and a valuation $\outval \in \oalphabet$ of the output propositions. For any letter $\sigma$, we consider the projection to input propositions by $\inpval \defeq \sigma \cap \inpv$ and to output propositions by $\outval \defeq \sigma \cap \outv$.
If the set $S$ is finite, $\trans$ is a finite-state transition system and its size is defined by $|\trans| \defeq |S|$. 

An \emph{execution} of $\trans$ is an infinite sequence $s_0, (\inpval_0 \cup \outval_0), s_1, (\inpval_1 \cup \outval_1), s_2\ldots$ such that $s_0$ is the initial state, and $(s_{i+1},\outval_i) = \tau(s_i,\inpval_i)$ for every $i \geq 0$. The corresponding sequence $(\inpval_0 \cup \outval_0),(\inpval_1 \cup \outval_1),\ldots \in \alphabet^\omega$ is called a trace. We denote with $\Traces(\trans)$ the set of all traces of a transition system $\trans$.
A transition system $\trans$ satisfies an LTL formula $\varphi$ over atomic propositions $\props = \inpv \cup \outv$, denoted by $\trans \models \varphi$, if for every $w \in \Traces(\trans)$, it holds that $w \models \varphi$.

The decision problem of determining whether there exists a transition system that satisfies an LTL formula is called the \emph{realizability problem for LTL}. If an LTL formula $\varphi$ is realizable, the goal of \emph{LTL synthesis problem} is to construct a transition system $\trans$ such that $\trans \models \varphi$.

\subsection{Bounded Synthesis Approach}\label{sec:def-boundedsynth}

The \emph{run graph} of a universal automaton $\autA = (Q,q_0,\delta,F)$ on a transition system $\trans = (S,s_0,\tau)$ is the unique graph $G = (V,E)$ with a set of nodes $V = S \times Q$ and a set of labeled edges $E \subseteq V \times \alphabet \times V$ such that $((s,q),\sigma,(s',q')) \in E$ iff $(q,\sigma,q') \in \delta$ and $\tau(s,\sigma\cap \inpv) = (s',\sigma\cap \outv)$.
That is, $G$ is the product of $\autA$ and $\trans$.

A run graph of a universal B\"uchi (resp.\ co-B\"uchi) automaton is accepting if every infinite path $(s_0,q_0),(s_1,q_1), \ldots$ contains infinitely (resp.\ finitely) many occurrences of states $q_i$ in $F$. A transition system $\trans$ is accepted by a universal automaton $\autA$ if the unique run graph of $\autA$ on $\trans$ is accepting. We denote with  $\lang\autA$ the set of transition systems accepted by $\autA$.\looseness=-1

The bounded synthesis approach is based on the fact that for every LTL formula $\varphi$ one can construct a universal co-B\"uchi automaton $\autA_\varphi$ with at most $2^{O(|\varphi|)}$ states such that $\trans \in \lang{\autA_\varphi}$ iff $\trans \models \varphi$, for every transition system $\trans$~\cite{KupfermanV05}.

An \emph{annotation} of a transition system $\trans = (S,s_0,\tau)$  with respect to a universal co-B\"uchi automaton $\autA = (Q,q_0,\delta,F)$ is a function $\anot : S \times Q \to \natsbot$ that maps nodes of the run graph of $\autA$ on $\trans$ to the set $\natsbot$. Intuitively, such an annotation is valid if every node $(s,q)$ that is reachable from the node $(s_0,q_0)$ is annotated with a natural number, which is an upper bound on the number of rejecting states visited on any path from $(s_0,q_0)$ to $(s,q)$. 
Formally, an annotation $\anot : S \times Q \to \natsbot$ is \emph{valid} if
\begin{compactitem}
	\item $\anot(s_0,q_0) \neq \bot$, i.e., the pair of initial states is labeled with a number, and 
	\item whenever $\anot(s,q) \neq \bot$, then for every edge $((s,q),\sigma,(s',q'))$ in the run graph of $\autA$ on $\trans$ we have that $(s',q')$ is annotated with a number (i.e., $\anot(s',q')\neq \bot$), such that
	$\anot(s',q') \geq  \anot(s,q)$, and if $q' \in F$, then $\anot(s',q') >  \anot(s,q)$.
\end{compactitem}
Valid annotations of finite-state transition systems correspond to accepting run graphs. An annotation $\anot$ is $c$-bounded if $\anot(s,q) \in \{0,\ldots,c\}\cup\{\bot\}$ for all $s \in S$ and $q \in Q$.

The synthesis method proposed in \cite{ScheweF07a,FinkbeinerS13} employs the following result in order to reduce the bounded synthesis problem to checking the satisfiability of propositional formulas: a transition system $\trans$ is accepted by a universal co-B\"uchi automaton $\autA = (Q,q_0,\delta,F)$ iff there exists a $(|\trans|\cdot|F|)$-bounded valid annotation for $\trans$ and $\autA$. One can estimate a bound on the size of the transition system, which allows to reduce the synthesis problem to its bounded version. Namely, if there exists a transition system that satisfies an LTL formula $\varphi$, then there exists a transition system satisfying $\varphi$ with at most $\big(2^{(|\subf\varphi| +\log |\varphi|)}\big)!^2$ states.

Let $\autA = (Q,q_0,\delta,F)$ be a universal co-B\"uchi automaton for the LTL formula $\spec$. Given a bound $b$ on the size of the desired transition system $\trans$, the bounded synthesis problem can be encoded as a satisfiability problem with the following sets of propositional variables and constraints.

{\bf Variables:} The variables represent the desired transition system $\trans$, and the desired valid annotation $\anot$ of the  run graph of $\autA$ on $\trans$. A transition system with $b$ states $S = \{1,\ldots,b\}$ is represented by Boolean variables $\tau_{s,\inpval,s'}$ and $o_{s,\inpval}$ for every $s, s' \in S$, $\inpval \in \ialphabet$, and $o \in \outv$. The variable $\tau_{s,\inpval,s'}$ encodes the existence of transition from $s$ to $s'$ on input $\inpval$, and the variable $o_{s,\inpval}$ encodes $o$ being true in the output from state $s$ on input $\inpval$.

The annotation $\anot$ is  represented by the following variables. For each $s\in S$ and $q \in Q$, there is a Boolean variable $\anotb_{s,q}$ and a vector $\anotn_{s,q}$ of $\log(b\cdot |F|)$ Boolean variables: the variable $\anotb_{s,q}$ encodes the reachability of $(s,q)$ from the initial node $(s_0,q_0)$ in the corresponding run graph, and the vector of variables $\anotn_{s,q}$ represents the bound for the node $(s,q)$. The constraints are as follows~\cite{FinkbeinerS13}.

{\bf Constraints for input-enabled $\trans$:}
$\quad C_\tau \defeq\bigwedge_{s \in S}\bigwedge_{\inpval \in \ialphabet}\bigvee_{s' \in S} \tau_{s,\inpval,s'}$.

{\bf Constraints for valid annotation:}\\
\begin{equation*}
	C_\anot  \defeq  
	\anotb_{s_0,q_0} \wedge 
	\bigwedge_{q,q' \in Q}\bigwedge_{s,s' \in S}\bigwedge_{\inpval \in \ialphabet}
	\Big( 
	\big(
	\anotb_{s,q} \wedge 
	\delta_{s,q,\inpval,q'} \wedge 
	\tau_{s,\inpval,s'}
	\big) \rightarrow 
	\succa_\anot(s,q,s',q')
	\Big),
\end{equation*}
where $\delta_{s,q,\inpval,q'}$ is a formula over the variables $o_{s,\inpval}$ that characterizes the transitions in $\autA$ between  $q$ and $q'$ on labels consistent with $\inpval$, and
$\succa_\anot(s,q,s',q')$ is a formula  over the annotation variables such that

\begin{equation*}
	\succa_\anot(s,q,s',q') \defeq 
	\begin{cases}
	\anotb_{s',q'} \wedge (\anotn_{s',q'} > \anotn_{s,q}) &\text{if } q' \in F,\\
	\anotb_{s',q'} \wedge (\anotn_{s',q'} \geq \anotn_{s,q}) &\text{if } q' \not\in F.
\end{cases}
\end{equation*}

\subsection{Maximum Satisfiability}\label{sec:def-maxsat}

The procedure proposed by Finkbeiner and Schewe~\cite{FinkbeinerS13} and recalled in the previous section provides a SAT encoding of synthesis when the size of the implementation is bounded. Maximum realizability is an optimization variant of synthesis while MaxSAT is an optimization variant of SAT. We show that, for a proposed value function, the maximum realizability problem under a bounded implementation size can be reduced to a \emph{partial weighted MaxSAT} instance.

Consider a propositional logic formula in conjunctive normal form (CNF), i.e., a formula that is a conjunction of disjunction of literals, where a literal is a Boolean variable or its negation and a disjunction of literals is called a clause.
\emph{MaxSAT} is the problem of assigning truth values to a set of Boolean variables such that the number of clauses of a propositional logic formula in CNF that are made true, is maximized \cite{biere2009handbook}. 
A \emph{partial MaxSAT} is a variant of MaxSAT problem where the clauses are categorized as hard and soft clauses. In this case, the goal is to find a truth assignment to the variables such that all the hard clauses are made true and the number of soft clauses that become true is maximized.
A more general problem is that of \emph{partial weighted MaxSAT} where each of the soft clauses is associated with a positive numerical weight. There, the objective is to find a truth assignment to the variables that not only makes all the hard clauses true but also maximizes the sum of the weights of the soft clauses that become true.

We exploit the separation of the hard and soft clauses in partial weighted MaxSAT to capture the hard and soft constraints that arise in the encoding of the maximum realizability problem. Furthermore, we design the weights of the soft clauses in a way to promote the quantitative objective associated with the conjunction of the given soft specifications.

\section{Maximum Realizability}\label{sec:prob}
Often, the specifications required from a system are a combination of multiple requirements, which might not be realizable in conjunction. In such a case, in addition to reporting the unrealizability to the system designer, we would like the synthesis procedure to construct an implementation that satisfies the specifications ``as much as possible''. Such implementation is particularly useful in the case where some of the requirements describe desirable but not necessarily essential properties of the system. To determine what ``as much as possible'' formally means, a quantitative semantics of the specification language is necessary. In the following, we provide such semantics for a fragment of LTL. The quantitative interpretation is based on the standard semantics of LTL formulas of the form $\gphi$.

\subsection{Quantitative Semantics of Soft Safety Specifications}\label{sec:quantitative-semantics}
Let $\gphi_1,\ldots,\gphi_n$ be a set of LTL specifications, where each $\softSpec_i$ is a safety LTL formula. In order to formalize the maximal satisfaction of $\gphi_1\wedge\ldots\wedge\gphi_n$, we first give a quantitative semantics of formulas of the form $\gphi$.

\paragraph{Quantitative semantics of safety specifications.} 
For an LTL formula of the form $\gphi$ and a transition system $\trans$, we define \emph{the value $\val\trans{\gphi}$ of $\gphi$ in $\trans$} as
$$
\val\trans\gphi \defeq
\begin{cases}
(1,1,1) & \text{if } \trans\models\LTLglobally\varphi,\\
(1,1,0) & \text{if } \trans\not\models\LTLglobally\varphi \text{ and } \trans\models\LTLfinally\LTLglobally\varphi,\\
(1,0,0) & \text{if } \trans\not\models\LTLglobally\varphi \text{ and } \trans\not\models\LTLfinally\LTLglobally\varphi \text{ and }\trans\models\LTLglobally\LTLfinally\varphi,\\
(0,0,0) & \text{if } \trans\not\models\LTLglobally\varphi \text{ and } \trans\not\models\LTLfinally\LTLglobally\varphi \text{ and }\trans\not\models\LTLglobally\LTLfinally\varphi.\\
\end{cases}
$$
Thus, the value of $\gphi$ in a transition system $\trans$ is a vector $(v_1,v_2,v_3) \in \{0,1\}^3$, where the value $(1,1,1)$ corresponds to the $\true$ value in the classical semantics of LTL. When $\trans\not\models\gphi$, the values $(1,1,0)$, $(1,0,0)$ and $(0,0,0)$ capture the extent to which $\varphi$ holds or not along the traces of $\trans$. For example, if $\val\trans\gphi =(1,0,0)$, then $\varphi$ holds infinitely often on each trace of $\trans$, but there exists a trace of $\trans$ on which $\varphi$ is violated infinitely often. When $\val\trans\gphi =(0,0,0)$, then on some trace of $\trans$, $\varphi$ holds for at most finitely many positions.
Note that by the definition of $\mathit{val}$, if $\val\trans\gphi = (v_1,v_2,v_3)$, then
\begin{itemize}
	\item $v_1=1$ if and only if $\trans\models \gfphi$,
	\item $v_2=1$ if and only if $\trans\models \fgphi$, 
	\item $v_3=1$ if and only if $\trans\models \gphi$.
\end{itemize}  
Thus, the lexicographic ordering on $\{0,1\}^3$ captures the preference of one transition system over another with respect to the  quantitative satisfaction of $\gphi$.

\begin{example}\label{ex:one-safety}
	Suppose that we want to synthesize a transition system representing a navigation strategy for a robot working at a restaurant. We require that the robot  serves the VIP area  infinitely often, formalized in LTL as $\LTLglobally\LTLfinally \mathit{vip\_area}$. We also desire that the robot never enters the staff's office, formalized as $\LTLglobally\neg\office$. Now, suppose that initially the key to the VIP area is in the office. Thus, in order to satisfy $\LTLglobally\LTLfinally \mathit{vip\_area}$, the robot must violate $\LTLglobally\neg\office$. A strategy in which the office is entered only once, and satisfies $\fg \neg\office$, is preferable to one which enters the office over and over again, and only satisfies $\gf\neg\office$.
	Thus, we want to synthesize a strategy $\trans$  maximizing $\val\trans{\LTLglobally\neg\office}$.
\end{example}

In order to compare implementations with respect to their satisfaction of a conjunction of several safety specifications $\gphi_1 \wedge \ldots \wedge \gphi_n$, we will extend the above definition. We first consider the case where the specifier has not expressed any preference for the individual conjuncts and later on, extend that to the case with a given priority ordering. Consider the following example.

\begin{example}\label{ex:two-safety} 
	We consider again the restaurant robot, now with two soft specifications. The soft specification $\LTLglobally (\mathit{req1}\rightarrow \LTLnext\mathit{table1})$ requires that each request by table  1 is served immediately at the next time instance. Similarly, $\LTLglobally (\mathit{req2}\rightarrow \LTLnext\mathit{table2})$, requires the same for table number 2. Since the robot cannot be at both tables simultaneously, formalized as the hard specification $\LTLglobally (\neg\mathit{table1} \vee\neg\mathit{table2})$, the conjunction of these requirements is unrealizable. Unless the two tables have priorities, it is preferable to satisfy each of $\mathit{req1}\rightarrow \LTLnext\mathit{table1}$ and $\mathit{req2}\rightarrow \LTLnext\mathit{table2}$ infinitely often, rather than serve one and the same table all the time. 
\end{example}

\paragraph{Quantitative semantics of conjunctions.} 
To capture the idea illustrated in Example~\ref{ex:two-safety}, we define a value function, which intuitively 
gives higher values to transition systems in which a fewer number of soft specifications have low values. Formally, let \emph{the value of $\gphi_1 \wedge \ldots \wedge \gphi_n$ in $\trans$} be
\[\val\trans{\gphi_1 \wedge \ldots \wedge \gphi_n} \defeq \big(
\sum_{i=1}^n v_{i,1},
\sum_{i=1}^n v_{i,2},
\sum_{i=1}^n v_{i,3}
\big),\] 
where
$\val\trans{\gphi_i} = (v_{i,1},v_{i,2},v_{i,3})$ for $i \in \{1,\ldots,n\}$. To compare transition systems according to these values, we use lexicographic ordering on $\{0,\ldots,n\}^3$.

\begin{example}\label{ex:two-safety-val}  
	For the specifications in Example~\ref{ex:two-safety}, the defined value function assigns value $(2,0,0)$ to a  system satisfying  $\gf(\mathit{req1}\rightarrow \LTLnext\mathit{table1})$ and $\gf (\mathit{req2}\rightarrow \LTLnext\mathit{table2})$, but  neither of $\fg  (\mathit{req1}\rightarrow \LTLnext\mathit{table1})$ and 
	$\fg (\mathit{req2}\rightarrow \LTLnext\mathit{table2})$. It assigns the smaller value 
	$(1,1,1)$
	to an implementation that gives priority to table 1 and satisfies $\LTLglobally (\mathit{req1}\rightarrow \LTLnext\mathit{table1})$ but not $\gf(\mathit{req2}\rightarrow \LTLnext\mathit{table2})$.
\end{example}

According to the definition above, a transition system that satisfies all soft requirements to some extent is considered better in the lexicographic ordering than a transition system that satisfies one of them exactly and violates all the others. We could instead inverse the order of the sums in the triple, thus giving preference to satisfying some soft specification exactly, over having some lower level of satisfaction over all of them. The next example illustrates the differences between the two variations.

\begin{example}\label{ex:ordering}
	For the two soft specifications from Example~\ref{ex:two-safety}, reversing the order of the sums in the definition of $\val\trans{\gphi_1 \wedge \ldots \wedge \gphi_n}$ results in giving the higher value $(1,1,1)$ to a transition system that satisfies $\LTLglobally (\mathit{req1}\rightarrow \LTLnext\mathit{table1})$ but not $\gf(\mathit{req2}\rightarrow \LTLnext\mathit{table2})$, and the lower value $(0,0,2)$ to the one that only guarantees $\gf(\mathit{req1}\rightarrow \LTLnext\mathit{table1})$ and $\gf (\mathit{req2}\rightarrow \LTLnext\mathit{table2})$. The most suitable ordering usually depends on the specific application. 
\end{example}

\subsection{Problem Formulation}\label{sec:prob-form}
Using the definition of quantitative satisfaction of soft safety specifications, we now define the maximum realizability problem, which asks to synthesize a transition system that satisfies a given \emph{hard} LTL specification, and is optimal with respect to the satisfaction of a conjunction of \emph{soft} safety specifications.

{\bf Maximum realizability problem:} Given an LTL formula $\spec$ and formulas $\gphi_1,\ldots,\gphi_n$, where each $\softSpec_i$ is a safety LTL formula, the maximum realizability problem asks to determine if there exists a transition system $\trans$ such that $\trans \models \spec$, and if the answer is positive, to synthesize a transition system $\trans$ such that $\trans \models \spec$, and for every transition system $\trans'$ with $\trans'\models \spec$ it holds that $\val\trans{\gphi_1 \wedge \ldots \wedge \gphi_n} \geq \val{\trans'}{\gphi_1 \wedge \ldots \wedge \gphi_n}$.

{\bf Bounded maximum realizability problem:} Given an LTL formula $\spec$ and formulas $\gphi_1,\ldots,\gphi_n$, where each $\softSpec_i$ is a safety LTL formula, and a bound $b\in \nats_{>0}$, the bounded maximum realizability problem asks to determine if there exists a transition system $\trans$ with $|\trans| \leq b$ such that $\trans \models \spec$, and if the answer is positive, to synthesize a transition system $\trans$ such that $\trans \models \spec$, $|\trans| \leq b$ and for every transition system $\trans'$ with $\trans'\models \spec$ and $|\trans'| \leq b$, it holds that $\val\trans{\gphi_1 \wedge \ldots \wedge \gphi_n} \geq \val{\trans'}{\gphi_1 \wedge \ldots \wedge \gphi_n}$.

\section{Maximum Realizability as Iterative MaxSAT Solving}\label{sec:maxsat-encoding}
We now describe the proposed MaxSAT-based approach to maximum realizability. First, we establish an upper bound on the minimal size of an implementation that satisfies a given LTL specification $\varphi$ and maximizes the satisfaction of a conjunction of the soft specifications $\gphi_1, \ldots,\gphi_n$, according to the value function defined in Section~\ref{sec:quantitative-semantics}.
This bound can be used to reduce the maximum realizability problem to its bounded version, which we encode as a MaxSAT problem.

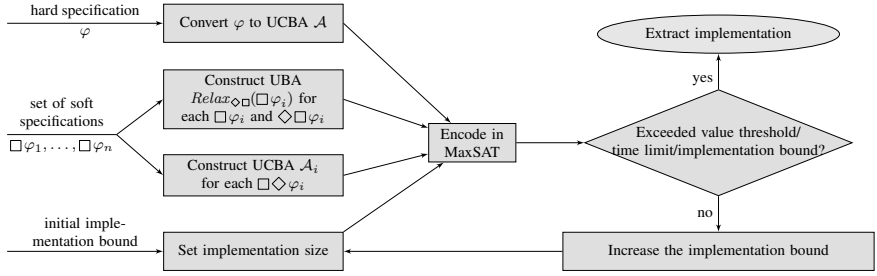
\begin{figure}[t]
	\centering
	\scalebox{.72}{
	\tikzstyle{operation} = [rectangle, draw, fill=gray!25,
text width=11em, text centered,  minimum height=2.5em]
\tikzstyle{condition} = [diamond, aspect=2.5, draw, fill=gray!25, 
text width=21em, text centered, node distance=3cm, inner sep=-10pt]
\tikzstyle{terminal} = [draw, ellipse, fill=gray!20, node distance=3cm, minimum height=2.5em, text width=11em, text centered, inner sep=0.5pt, node distance=2cm]
\tikzstyle{line} = [draw, -latex']

\begin{tikzpicture}[node distance = 1.5cm, auto]
\tikzstyle{every node}=[font=\footnotesize]

\node [operation] (ucba_h) {Convert $\varphi$ to UCBA $\autA$};
\node [operation, below of=ucba_h,yshift=.1cm] (uba_s) {Construct UBA $\relaxfg(\gphi_i)$ for each $\LTLglobally \varphi_i$ and $\LTLfinally \LTLglobally \varphi_i$};
\node [operation, below of=uba_s,yshift=.1cm] (ucba_s) {Construct UCBA $\autA_i$ for each $\LTLglobally \LTLfinally \varphi_i$};
\node [operation, below of=ucba_s,yshift=.1cm] (set_b) {Set implementation size};
\node [operation, right of=uba_s, node distance=4cm,yshift=-.8cm,text width=5em] (maxsat) {Encode in MaxSAT};
\node [condition, right of=maxsat, node distance=4.5cm] (term_cond) {Exceeded value threshold/\\time limit/implementation bound?};
\node [terminal, above of=term_cond, node distance=2cm] (extract) {Extract implementation};
\node [operation, below of=term_cond, node distance=2cm,text width=5.5cm] (inc_b) {Increase the implementation bound};

\path [line] (ucba_h.east) -- (maxsat);
\path [line] (uba_s.east) -- (maxsat);
\path [line] (ucba_s.east) -- (maxsat);
\path [line] (maxsat) -- (term_cond);
\path [line, -] +(-4.5cm,-2.05cm) -- node[above, text centered, text width=2cm] {set of soft specifications} node[below] {$\LTLglobally\varphi_1, \ldots, \LTLglobally\varphi_n$} +(-2.5cm,-2.05cm);
\path [line] +(-4.5cm,0cm) -- node[above] {hard specification} node[below] {$\varphi$} (ucba_h);
\path [line] +(-2.5cm,-2.05cm) -- (uba_s.west);
\path [line] +(-2.5cm,-2.05cm) -- (ucba_s.west);
\path [line] +(-4.5cm,-4.2cm) -- node[above, text centered, text width=3cm] {initial implementation bound}(set_b);
\path [line] (set_b.north east) -- (maxsat);
\path [line] (term_cond) -- node[left,near start] {yes} (extract);
\path [line] (term_cond) -- node[left] {no} (inc_b);
\path [line] (inc_b) -- (set_b);

\end{tikzpicture}
	}
	\caption{Outline of the maximum realizability procedure.}
	\label{fig:proc_diag}
\end{figure}

\subsection{Bounded Maximum Realizability}
To establish an upper bound on the minimal (in terms of size) optimal implementation, we make use of an important property of the function $\mathit{val}$ defined in Section~\ref{sec:quantitative-semantics}. Namely, the property that for each of the possible values of $\gphi_1 \wedge\ldots \wedge\gphi_n$  there is a corresponding LTL formula that encodes this value in the classical LTL semantics, as we formally state in the next lemma.

\begin{lemma}\label{lem:value-as-ltl}
For every transition system $\trans$ and soft safety specifications $\gphi_1,\ldots,$ $\gphi_n$, if $\val\trans{\gphi_1 \wedge \ldots \wedge \gphi_n} = v$, then there exists an LTL formula $\psi_v$ where
\begin{compactitem}
\item[(1)] $\psi_v = \softSpec_1'\wedge\ldots\wedge\softSpec_n'$, such that $\softSpec_i' \in\{\gphi_i,\fgphi_i,\gfphi_i,\true\} \text{ for }i=1,\ldots,n$, 
\item[(2)] $\trans \models \psi_v$, and for every $\trans'$, if $\trans' \models \psi_v$, then $\val{\trans'}{\gphi_1 \wedge \ldots \wedge \gphi_n} \geq v$.
\end{compactitem}
\end{lemma}
\begin{proof}
For each $i \in \{1,\ldots,n\}$, let $(v_{i,1},v_{i,2},v_{i,3}) = \val\trans{\gphi_i}$, and let
\[\psi_v^i \defeq
\begin{cases}
\gphi_i & \text{if } v_{i,3} = 1,\\
\fgphi_i & \text{if } v_{i,3} = 0 \text{ and } v_{i,2} = 1,\\
\gfphi_i  & \text{if } v_{i,2} = 0 \text{ and } v_{i,1} = 1,\\
\true & \text{if } v_{i,1} = 0.\\
\end{cases}
\]

We define $\psi_v  = \bigwedge_{i=1}^n\psi_v^i$. By the definition of $\val\trans{\gphi_i}$ and $\psi_v^i$, we have that $\trans \models \psi_v^i$ for all $i\in \{1,\ldots,n\}$. Thus, we  conclude that $\trans \models \psi_v$. Clearly, $\psi_v$ also satisfies condition \emph{(1)}. Now, consider a transition system $\trans'$ with $\trans' \models \psi_v$.

Let $(v_1',v_2',v_3') = \val{\trans'}{\gphi_1 \wedge \ldots \wedge \gphi_n}$. We will show that $v_1' \geq v_1$, $v_2' \geq v_2$ and $v_3' \geq v_3$, where $(v_1,v_2,v_3) = v$.  Fix some $i \in\{1,\ldots,n\}$.

Let $(v_{i,1}',v_{i,2}',v_{i,3}') = \val{\trans'}{\gphi_i}$. Since $\trans' \models \psi_v$ we have that $\trans' \models \psi_v^i$. Thus by the definition of $\psi_v^i$, we have that if $v_{i,3} = 1$, then $\trans' \models \gphi_i$, and thus $v_{i,3}'=1$. Similarly, if $v_{i,2} = 1$ we can conclude that $v_{i,2}' = 1$, and if  
$v_{i,1} = 1$, then we have $v_{i,1}' = 1$. Since $i \in \{1,\ldots,n\}$ was arbitrary, and since 
\[
\begin{array}{lll}
(v_1,v_2,v_3) & = &  \big(
\sum_{i=1}^n v_{i,1},
\sum_{i=1}^n v_{i,2},
\sum_{i=1}^n v_{i,3}
\big)\text{ and }\\
(v_1',v_2',v_3') & = &  \big(
\sum_{i=1}^n v_{i,1}',
\sum_{i=1}^n v_{i,2}',
\sum_{i=1}^n v_{i,3}'
\big),
\end{array}
\]
we can conclude that $v_1' \geq v_1$, $v_2' \geq v_2$ and $v_3' \geq v_3$. This implies that $(v_1',v_2',v_3') \geq (v_1,v_2,v_3)$ also according to the lexicographic ordering, which proves \emph{(2)}.\qed
\end{proof}

The following theorem is a consequence of Lemma~\ref{lem:value-as-ltl}. 

\begin{theorem}\label{thm:optimal-bound-safety}
Given an LTL specification $\spec$ and soft safety specifications $\gphi_1,\ldots,$ $\gphi_n$, 
if there exists a transition system $\trans \models \spec$, then there exists  $\trans^*$ such that
\begin{compactitem}
\item [(1)] $\val{\trans^*}{\gphi_1 \wedge \ldots \wedge \gphi_n} \geq\val{\trans}{\gphi_1 \wedge \ldots \wedge \gphi_n}$ for all $\trans$ with $\trans \models \spec$,
\item[(2)] $\trans^* \models \spec$ and $|\trans^*| \leq \left((2^{(b+\log b)})!\right)^2$,
\end{compactitem}
 where $b = \max\{|\subf{\spec\wedge\softSpec_1'\wedge\ldots\wedge\softSpec_n'}| \mid \forall i:\ \softSpec_i' \in\{\gphi_i,\fgphi_i,\gfphi_i\}\}$.
\end{theorem}
\begin{proof}
	Let $v^* = \max\big\{v \in \{0,\ldots,n\}^3 \mid \exists \trans: \trans\models \varphi \text{ and } \val\trans{\gphi_1 \wedge \ldots \wedge \gphi_n} = v\big\}.$
	Let $\trans$ be a transition system such that $\trans \models \varphi$ and 
	$\val\trans{\gphi_1 \wedge \ldots \wedge \gphi_n} = v^*$. According to Lemma~\ref{lem:value-as-ltl}, there exists an LTL formula $\psi_{v^*}$ that satisfies the conditions of the lemma. Thus, $\trans \models \varphi \wedge \psi_{v^*}$. According to~\cite{ScheweF07a}, there exists a transition system $\trans^*$ such that $\trans^* \models \varphi \wedge \psi_{v^*}$ and $|\trans^*| \leq \left(\big(2^{|\subf{\varphi \wedge \psi_{v^*}}| +\log |\varphi \wedge \psi_{v^*}|}\big)!\right)^2$. Combining this with the guarantees of Lemma~\ref{lem:value-as-ltl}, we get that $\val{\trans^*}{\gphi_1 \wedge \ldots \wedge \gphi_n} \geq v^*$, 
	$\trans^* \models \spec$ and $|\trans^*| \leq \left((2^{b+\log b})!\right)^2$.
	Thus, $\trans^*$ satisfies condition \emph{(2)}, and since the value $v^*$ is optimal, we have that condition \emph{(1)} holds as well. \qed
\end{proof}

The bound above is estimated based on the size of the specifications, using a worst-case bound on the size of the corresponding automata. Given the automata for all the specifications $\gphi_i,\fgphi_i$ and $\gfphi_i$, a potentially better bound can be estimated based on the sizes of these automata.

Lemma~\ref{lem:value-as-ltl} immediately provides a naive synthesis procedure, which searches for an optimal implementation by enumerating possible $\psi_v$ formulas and solving the corresponding realizability questions. The total number of these formulas is $4^n$, where $n$ is the number of soft specifications. The approach that we propose avoids this rapid growth, by reducing the optimization problem to a single MaxSAT instance, making use of  the power of the state-of-the-art MaxSAT solvers.

Figure~\ref{fig:proc_diag} gives an overview of our maximum realizability procedure and the automata constructions it involves. As in the bounded synthesis approach, we construct a universal co-B\"uchi automaton $\autA$ for the hard specification $\varphi$.  For each soft specification $\gphi_j$, we construct a pair of automata corresponding to the relaxations of $\gphi_j$. The relaxation $\gfphi_j$ is treated as in bounded synthesis. For $\gphi_i$ and $\fgphi_i$, we construct a single universal B\"uchi automaton and define a corresponding annotation function as described next.

\subsection{Automata and Annotations for Soft Safety Specifications}\label{sec:automata-safety}

We present here the reduction to MaxSAT for the case when each soft specification is of the form $\gpsi$ where $\psi$ is a \emph{syntactically safe} LTL formula. In this case, we construct a single automaton for both $\gpsi$ and its relaxation $\fgpsi$,  and encode the existence of a single annotation function in the MaxSAT problem. The size of this automaton is at most exponential in the size of $\gpsi$.
In the general case, we can treat $\gpsi$ and $\fgpsi$ separately, in the same way that we treat the relaxation $\gfpsi$  of $\gpsi$ in the presented encoding. That would require in total three instead of two annotation functions per soft specification.

We now describe the construction of a universal B\"uchi automaton  $\autB_{\scriptsize \gpsi}$ for a syntactically safe soft specification $\gpsi$ and show how we can modify it to obtain an automaton $\relaxfg(\gpsi)$ that incorporates the relaxation of $\gpsi$ to $\fgpsi$.

\begin{proposition}\label{prop:aut-globally}
	Given an LTL formula $\gpsi$ where $\psi$ is syntactically safe, we can construct a universal B\"uchi automaton  $\autB_{\scriptsize \gpsi} = (Q_{\scriptsize \gpsi},q_0^{\scriptsize \gpsi},\delta_{\scriptsize \gpsi},F_{\scriptsize \gpsi})$ such that 
	$\lang{\autB_{\scriptsize \gpsi}} = \{\trans \mid \trans \models \gpsi\}$, and $\autB_{\scriptsize \gpsi}$ has a unique non-accepting sink state, that is, there exists a unique state $\rej_\psi \in Q_{\scriptsize \gpsi}$ such that 
	$F_{\scriptsize \gpsi} = Q_{\scriptsize \gpsi} \setminus \{\rej_\psi\}$, and for every $\sigma \in \Sigma$ it holds that $\{q \in Q_{\scriptsize \gpsi} \mid (\rej_\psi,\sigma,q) \in \delta_{\scriptsize \gpsi}\}  = \{\rej_\psi\}$.
\end{proposition}
\begin{proof} We first describe the construction of the automaton $\autB_{\scriptsize \gpsi}$ of the desired form, and then proceed to prove its correctness.

{\bf Construction.}
First we construct a universal B\"uchi automaton $\autB_\psi$ for the formula $\psi$ such that, for every word $w \in \alphabet^\omega$, it holds that $w$ is accepted by $\autB_\psi$ if and only if $w \models \psi$. To this end, we use the following result from~\cite{KupfermanV01}.
Given a syntactically safe LTL formula $\psi$, we 
can construct a nondeterministic finite automaton $\autN = (Q_{\autN},q_0^{\autN},\delta_{\autN},F_{\autN})$ with at most $2^{\bigo{|\psi|}}$ states, and such that:\\
-- if $v \in \alphabet^*$ is accepted by $\autN$, then for all $w' \in \alphabet^\omega$ we have $vw' \not\models \psi$, and\\
-- for every $w \in \alphabet^\omega$, if $w \not \models \psi$, then there exists a prefix $v$ of $w$ accepted by $\autN$.
Thus, $\autN$ accepts at least one bad prefix of each word $w \in \alphabet^\omega$ that violates $\psi$. 
 
The automaton $\autB_\psi = (Q_\psi,q_0^\psi,\delta_\psi,F_\psi)$ is obtained from $\autN$ as follows. The set of states of $\autB_\psi$ consists of those states of $\autN$ that are not accepting, together with a new state $\rej_\psi \not\in Q_{\autN}$, that is, $Q_\psi = (Q_{\autN} \setminus F_{\autN}) \cup \{\rej_\psi\}$. We let $q_0^\psi = q_0^{\autN}$ and $F_\psi = Q_\psi \setminus \{\rej_\psi\}$. The transition relation of $\autB_\psi$ is obtained from $\delta_{\autN}$ by redirecting all transitions leading to states in $F_{\autN}$ to the new state $\rej_\psi$. Formally, 
\begin{align*}
\delta_\psi = \left(\delta_{\autN} \cap (Q_\psi \times \Sigma \times Q_\psi)\right)& \cup \{(q,\sigma,\rej_\psi) \mid \sigma \in \alphabet \text{ and }\exists q' \in F_{\autN}.\ (q,\sigma,q') \in \delta_{\autN}\}\\&
\cup \{(\rej_\psi,\sigma,\rej_\psi) \mid \sigma \in \alphabet\}.
\end{align*}

\noindent
We now construct a universal B\"uchi automaton $\autB_{\scriptsize \gpsi} = (Q_{\scriptsize \gpsi},q_0^{\scriptsize \gpsi},\delta_{\scriptsize \gpsi},F_{\scriptsize \gpsi})$ such that $w$ is accepted by $\autB_{\scriptsize \gpsi}$ iff $w \models \gpsi$.
We let $Q_{\scriptsize \gpsi} = Q_\psi$, $q_0^{\scriptsize \gpsi} = q_0^\psi$, and  $F_{\scriptsize \gpsi} = F_\psi$. The transition relation $\delta_{\scriptsize \gpsi}$ extends $\delta_\psi$ by adding a self-loop at the initial state $q_0^\gpsi$ for all transitions from $q_0^\psi$ in $\delta_\psi$ that do not lead to $\rej_\psi$:
\begin{align*}
\delta_{\scriptsize \gpsi} = \delta_\psi & \cup  \{(q_0^\psi,\sigma,q_0^\psi) \mid \sigma \in \alphabet \text{ and }\exists q' \in (Q_\psi \setminus \{\rej_\psi\}).\ (q_0^\psi,\sigma,q') \in \delta_\psi\}.
\end{align*}

{\bf Correctness.}
Let  $\trans \in \lang{\autB_{\scriptsize \gpsi}}$. 
Since $\autB_{\scriptsize \gpsi}$ is a universal B\"uchi automaton, this means that the unique run graph of $\autB_{\scriptsize \gpsi}$ on $\trans$ is accepting, which in turn means that each infinite path contains infinitely many occurrences of states in $F_{\scriptsize \gpsi}$. Since $F_{\scriptsize \gpsi}$ contains all states except $\rej_\psi$, and $\rej_\psi$ is a sink state, it follows that every infinite path in the run graph contains only states in 
$F_{\scriptsize \gpsi}$. 

Suppose, for the sake of contradiction, that $\trans\not\models\gpsi$. Thus, there exists $\omega = \sigma_0,\sigma_1,\ldots\in\Traces(\trans)$ such that $\omega\not\models\gpsi$. Let $i \geq 0$ be an index such that $\sigma_i,\sigma_{i+1},\ldots \not\models\psi$. By the choice of the automaton $\autN$, there exists a prefix of $\sigma_i,\sigma_{i+1},\ldots$ accepted by $\autN$.
Since $\omega \in \Traces(\trans)$ and $\trans \in \lang{\autB_{\scriptsize \gpsi}}$, every path in the run graph corresponding to $\omega$ never visits $\rej_\psi$. Thus, since $\delta_{\scriptsize \gpsi}$ contains a self-loop at state $q_0^{\scriptsize \gpsi}$ with letters not leading to $\rej_\psi$, there exists a path in the run graph corresponding to $\sigma_0,\ldots,\sigma_{i-1}$ that ends in $q_0^{\scriptsize \gpsi}$. By the definition of $\delta_{\psi}$ and the existence of an accepting run of $\autN$ on a prefix of $\sigma_i,\sigma_{i+1},\ldots$ we can conclude that there exists a path in the run graph of $\autB_{\scriptsize \gpsi}$ corresponding to $\omega$ that reaches $\rej_\psi$, which is a contradiction.

For the other direction, consider a transition system $\trans$ such that $\trans \models \gpsi$, and suppose that $\trans \not \in \lang{\autB_{\scriptsize \gpsi}}$. This means that there exists an infinite path in the run graph of $\autB_{\scriptsize \gpsi}$ on $\trans$ that visits states in $F_{\scriptsize \gpsi}$ only finitely many times, which means that this path eventually reaches $\rej_{\psi}$. Let $\omega = \sigma_0,\sigma_1,\ldots\in\Traces(\trans)$ be the word corresponding to this path, and $i\geq 0$ be the last occurrence of $q_0^{\scriptsize \gpsi}$ on this path and $j > i$ be the index of the first occurrence of $\rej_\psi$.
Due to the definition of $\delta_\psi$, this implies that there exists an accepting run of $\autN$ on the word $\sigma_i,\ldots,\sigma_{j-1}$. Thus, $\sigma_i,\sigma_{i+1},\ldots\not\models\psi$, which in turn means  that $\omega\not\models\gpsi$. This is a contradiction with $\trans \models \gpsi$, and thus we can conclude that $\trans \in \lang{\autB_{\scriptsize \gpsi}}$.\qed
\end{proof}

From $\autB_{\scriptsize \gpsi}$, which has at most $2^{\mathcal{O}(|\psi|)}$ states, we obtain a universal B\"uchi automaton $\relaxfg(\gpsi)$ constructed by redirecting all the transitions leading  to $\rej_\psi$ to the initial state $q_0^{\scriptsize \gpsi}$. Formally, 
$\relaxfg(\gpsi) = (Q,q_0,\delta,F)$, where $Q = Q_{\scriptsize \gpsi} \setminus \{\rej_\psi\}$, $q_0 = q_0^{\scriptsize \gpsi}$, $F = F_{\scriptsize \gpsi}$ and 
$\delta = \big(\delta_{\scriptsize \gpsi} \setminus \{(q,\sigma,q') \in \delta_{\scriptsize \gpsi}\mid q' = \rej_\psi\}\big) \cup \{(q,\sigma,q_0) \mid (q,\sigma,\rej_\psi)\in \delta_{\scriptsize \gpsi}\}.$

Let $\Rej(\relaxfg(\gpsi)) = \{(q,\sigma,q_0) \in \delta \mid (q,\sigma,\rej_\psi) \in \delta_{\scriptsize \psi}\}$ be the set of transitions in $\relaxfg(\gpsi)$ that correspond to transitions in $\autB_{\scriptsize \gpsi}$ leading to $\rej_\psi$. 
The  automaton $\relaxfg(\gpsi)$ has the property that its run graph on a transition system $\trans$  does \emph{not} contain a reachable edge corresponding to a transition in $\Rej(\relaxfg(\gpsi))$ iff $\trans$ is accepted by the automaton $\autB_{\scriptsize \gpsi}$, (i.e., $\trans\models\gpsi$). 
Otherwise, if the run graph of $\relaxfg(\gpsi)$ on $\trans$ contains a reachable edge that belongs to $\Rej(\relaxfg(\gpsi))$, then  $\trans\not\models\gpsi$. However, if each infinite path in the run graph contains only a finite number of occurrences of such edges, then $\trans\models\fgpsi$. Based on these observations, we define an annotation function that annotates each node in the run graph with an upper bound on the number of edges in $\Rej(\relaxfg(\gpsi))$ visited on any path reaching the node. 

The next proposition formalizes the property that a transition system $\trans$ is accepted by $\autB_{\scriptsize \gpsi}$ if and only if the run graph of $\relaxfg(\gpsi)$ on $\trans$ does not contain a reachable edge corresponding to a transition in $\Rej(\relaxfg(\gpsi))$. 
\begin{proposition}\label{prop:rej-trans}\sloppy
Let $\trans$ be a transition system and let $G = (V,E)$ be the run graph of $\relaxfg(\gpsi)$ on $\trans$. Then, $\trans \in\lang{\autB_{\scriptsize \gpsi}}$ iff for every  $((s,q),\sigma,(s',q')) \in E$ with $(q,\sigma,q') \in  \Rej(\relaxfg(\gpsi))$, $(s,q)$ is not reachable from $(s_0,q_0)$ in $G$.
\end{proposition}
\begin{proof}
Suppose, for the sake of contradiction, that $\trans \in\lang{\autB_{\scriptsize \gpsi}}$ and let $G'=(V',E')$ be the run graph of $\autB_{\scriptsize \gpsi}$ on $\trans$. Suppose that there exists a path $(s_0,q_0),\sigma_0,\ldots (s_l,q_l)$ in $G$ such that there exists an edge $((s_l,q_l),\sigma,(s',q')) \in E$ with $(q_l,\sigma,q') \in  \Rej(\relaxfg(\gpsi))$. Without loss of generality, we assume that 
$(q_i,\sigma_i,q_{i+1}) \not \in \Rej(\relaxfg(\gpsi))$ for all $i=0,\ldots,l-1$. Then, the sequence $(s_0,q_0),\sigma_0,\ldots (s_l,q_l),\sigma,(s',q')$ corresponds to a path in the run graph $G'$ of $\autB_{\scriptsize \gpsi}$ on $\trans$ which enters the state $\rej_\psi$. Since $\rej_\psi$ is a non-accepting sink state, we conclude that $G'$ is not accepting. This implies $\trans \not\in\lang{\autB_{\scriptsize \gpsi}}$, which is a contradiction.

\sloppy Suppose now that for every node $(s,q)$ reachable in $G$ from $(s_0,q_0)$ and every edge $((s,q),\sigma,(s',q')) \in E$ we have that $(q,\sigma,q') \not\in  \Rej(\relaxfg(\gpsi))$. 
Assume that $\trans \not\in\lang{\autB_{\scriptsize \gpsi}}$, which means that there exists an infinite path from $(s_0,q_0)$ in the run graph of $\autB_{\scriptsize \gpsi}$ on $\trans$ that reaches the state $\rej_\psi$. This path corresponds to a path in $G$ from $(s_0,q_0)$ to some state $(s,q)$ for which there is an edge  $((s,q),\sigma,(s',q')) \in E$ with $(q,\sigma,q') \in  \Rej(\relaxfg(\gpsi))$, which is a contradiction.\qed
\end{proof}

A function $\anotfg : S \times Q \to \natsbot$ is \emph{\fgvalid}\ annotation for the run graph of the automaton $\relaxfg(\gpsi) = (Q,q_0,\delta,F)$ on the transition system $\trans = (S,s_0,\tau)$ if 
\begin{compactitem}
\item[\emph{(1)}] $\anotfg(s_0,q_0) \neq \bot$, i.e., the pair of initial states is labeled with a number, and 
\item[\emph{(2)}] if $\anotfg(s,q) \neq \bot$, then for every edge $((s,q),\sigma,(s',q'))$ in the run graph, we have that $\anotfg(s',q') \neq \bot$, and
\begin{compactitem}
\item if $(q,\sigma,q') \in\Rej(\relaxfg(\gpsi))$, then $\anotfg(s',q') >  \anotfg(s,q)$, and 
\item if $(q,\sigma,q') \not\in\Rej(\relaxfg(\gpsi))$, then $\anotfg(s',q') \geq  \anotfg(s,q)$.
\end{compactitem}
\end{compactitem}
This guarantees that $\trans \models \fgpsi$ iff there exists a \fgvalid\ $|\trans|$-bounded annotation $\anotfg$ for $\trans$ and $\relaxfg(\gpsi)$. Moreover, if $\anotfg$ is $|\trans|$-bounded and $\anotfg(s_0,q_0) = |\trans|$, then $\trans \models \gpsi$, as this means that no edge in $\Rej(\relaxfg(\gpsi))$ is ever reached.

\begin{proposition}\label{prop:anotfg}
Let $\trans = (S,s_0,\tau)$ be a finite-state transition system, and $G = (V,E)$ be the run graph of  $\relaxfg(\gpsi)$ on $\trans$. Then, $\trans \models \fgpsi$ if and only if there exists a \fgvalid\ $|\trans|$-bounded annotation for $G$. 
\end{proposition}
\begin{proof}\sloppy
Suppose that $\trans \models \fgpsi$. We will fist show that in every infinite path from $(s_0,q_0)$ in $G$ there are at most $|S|$ occurrences of edges whose corresponding transitions are in $\Rej(\relaxfg(\gpsi))$, and then we will use this fact to define a \fgvalid\ $|S|$-bounded annotation. Assume, for the sake of contradiction, that there exists an infinite path $(s_0,q_0),\sigma_0,(s_1,q_1),\sigma_1,\ldots$ such that for infinitely many positions $i\geq 0$ it holds that $(q_i,\sigma_i,q_{i+1}) \in \Rej(\relaxfg(\gpsi))$. Let $i_1 < i_2 <\ldots$ be a sequence of such positions. By the construction of $\relaxfg(\gpsi)$, we have $q_{i_j+1} = q_0$ for each $i_j$. Thus, using reasoning similar to that in Proposition~\ref{prop:rej-trans}, we can show that the trace $\sigma_0,\sigma_1,$ contains infinitely many positions $k$ such that $\sigma_k,\sigma_{k+1},\ldots \not\models\psi$. This means that $\sigma_0,\sigma_1,\ldots \not \models \fgpsi$. Since $\sigma_0,\sigma_1,\ldots \in \Traces(\trans)$, we can conclude that $\trans \not \models \fgpsi$, which is a contradiction. 

Thus, each infinite path in G contains only finitely many occurrences of edges in $\Rej(\relaxfg(\gpsi))$. Since the number of distinct nodes in $G$ of the form $(s,q_0)$ is $|S|$, we obtain an upper bound of $|S|$ occurrences of transitions from $\Rej(\relaxfg(\gpsi))$ on every path in $G$. Thus, we can construct a \fgvalid\ $|S|$-bounded annotation $\anotfg$ by mapping each reachable node $(s,q)$ to the maximal number of transitions from $\Rej(\relaxfg(\gpsi))$ on a path from $(s_0,q_0)$ to $(s,q)$, and mapping each unreachable node to $\bot$.

For the other direction, suppose that $\anotfg$ is a \fgvalid\ $|S|$-bounded annotation for $\trans$ and $\relaxfg(\gpsi)$. Assume that $\trans \not\models \fgpsi$. This means that there exists a trace $w = \sigma_0,\sigma_1,\ldots \in \Traces(\trans)$ such that for every position $i$ it holds that $\sigma_i,\sigma_i+1,\ldots \not\models \gpsi$. Let $s_0,\sigma_0,s_1,\sigma_1\ldots$ be the execution of $\trans$ corresponding to $w$. Since $\sigma_i,\sigma_i+1,\ldots \not\models \gpsi$, with reasoning similar to Proposition~\ref{prop:rej-trans} we can establish that there exists a path in
$G$ starting from $(s_i,q_0)$ that eventually takes an edge corresponding to a transition in $\Rej(\relaxfg(\gpsi))$, and by the construction of $\relaxfg(\gpsi)$, this transition leads to the node $(s_j,q_0)$ for some $j > i$. Thus, by induction, we can establish the existence of an infinite path $(s_0,q_0),(s_1,q_1),\ldots$ in $G$ that contains infinitely many occurrences of edges whose transitions are in $\Rej(\relaxfg(\gpsi))$. Since the annotation $\anotfg$ is \fgvalid, we can show by induction that for each $i\geq 0$ it holds that $\anotfg(s_i,q_i) \in \nats$ and $\anotfg(s_{i+1},q_{i+1})  \geq \anotfg(s_i,q_i)$. Since $G$ is finite, this path contains an edge  $((s_i,q_i),\sigma_i,(s_{i+1},q_{i+1}))$ for which $(q_i,\sigma_i,q_{i+1})\in\Rej(\relaxfg(\gpsi))$, and which is such that there exists $j \leq i$ such that $(s_{i+1},q_{i+1}) = (s_j,q_j)$. Since the annotation $\anotfg$ is \fgvalid, we have that $\anotfg(s_j,q_j) \leq \anotfg(s_i,q_i)$ and $\anotfg(s_i,q_i) < \anotfg(s_{i+1},q_{i+1})$, which contradicts $(s_{i+1},q_{i+1}) = (s_j,q_j)$. Thus, by contradiction, we conclude that $\trans \models \fgpsi$. \qed
\end{proof}
In particular, we have that if $\anotfg(s_0,q_0) \in \nats$, then $\trans \models \fgpsi$, and if $\anotfg$ is $c$-bounded and $\anotfg(s_0,q_0) = c$, then $\trans \models \gpsi$.
This property allows us to capture the satisfaction of $\gpsi$ and $\fgpsi$ with soft clauses for the same annotation function in the MaxSAT formulation.

\subsection{MaxSAT Encoding of Bounded Maximum Realizability}\label{sec:maxsat}

Let $\autA = (Q,q_0,\delta,F)$ be a universal co-B\"uchi automaton for the LTL formula $\spec$.
For each syntactically safe formula $\gphij$, $j \in\{1,\ldots,n\}$, we consider two universal automata:
the universal automaton $\autB_j = \relaxfg(\gphi_j)= (Q_j,q_0^j,\delta_j,F_j)$, constructed as described in Section~\ref{sec:automata-safety}, and 
a universal co-B\"uchi automaton $\autA_j = (\widehat Q_j,\widehat q_0^j,\widehat\delta_j,\widehat F_j)$ for the formula $\gfphi_j$. 
Given a bound $b$ on the size of the desired transition system, we encode the bounded maximum realizability problem as a MaxSAT problem with the following sets of variables and constraints.

{\bf Variables:} The MaxSAT formulation includes the variables from the SAT formulation of the bounded synthesis problem, which represent the desired transition system $\trans$ and the desired valid annotation of the run graph of $\autA$ on $\trans$. Additionally, it includes variables for representing the annotations $\anotfg_j$ and $\anot_j$ for $\autB_j$ and $\autA_j$, respectively, similarly to $\anot$ in the SAT encoding. More precisely, the variables for $\anotfg_j$ and $\anot_j$ are respectively represented by variables $\anotfgbj_{s,q}$ and $\anotfgnj_{s,q}$  where $s\in S$ and $q \in Q_j$, and variables $\anotbj_{s,q}$ and $\anotnj_{s,q}$ where $s\in S$ and $q \in \widehat Q_j$.

The set of constraints includes $C_\tau$ and $C_\anot$ from the SAT formulation as hard constraints, as well as the following constraints for the new annotations.

{\bf Hard constraints for valid annotations:}
For each $j=1,\ldots,n$, let
\begin{align*}
C_\anotfg^{j}\defeq\bigwedge_{q,q' \in Q_j}\bigwedge_{s,s' \in S}\bigwedge_{\inpval \in \ialphabet}
\Big( &
\big(
\anotfgbj_{s,q} \wedge
\delta^j_{s,q,\inpval,q'}\wedge 
\tau_{s,\inpval,s'}
\big) \rightarrow 
\succa_{\anotfg}^j(s,q,s',q',\inpval) 
\Big),
\\
C_\anot^{j}\defeq\bigwedge_{q,q' \in \widehat Q_j}\bigwedge_{s,s' \in S}\bigwedge_{\inpval \in \ialphabet}
\Big( &
\big(
\anotbj_{s,q} \wedge
\widehat\delta^j_{s,q,\inpval,q'}\wedge 
\tau_{s,\inpval,s'}
\big) \rightarrow 
\succa_{\anot}^j(s,q,s',q',\inpval)
\Big),
\end{align*}

\begin{align*}
\text{where }\quad \succa_\anotfg^j(s,q,s',q',\inpval) \defeq \anotfgbj_{s',q'} \wedge &
\big(\rej^j(s,q,q',\inpval) \rightarrow \anotfgnj_{s',q'} > \anotfgnj_{s,q}\big) \wedge
\\&
\big(\neg\rej^j(s,q,q',\inpval) \rightarrow \anotfgnj_{s',q'} \geq \anotfgnj_{s,q}\big)
,
 \end{align*}
and $\rej^j(s,q,q',\inpval)$ is a formula over $o_{s,\inpval}$ obtained from $\Rej(\autB_j)$.\sloppy
The formula $\succa_{\anot}^j(s,\widehat q,s',\widehat q',\inpval)$ is analogous to 
$\succa_{\anot}(s,q,s',q',\inpval)$ defined in Section~\ref{sec:def-boundedsynth}.

{\bf Soft constraints for valid annotations:} Let $b \in \nats_{>0}$ be the bound on the size of the transition system.
For each $j=1,\ldots,n$, we define
\[
\begin{array}{lllll}
\soft_{\tiny\LTLglobally}^{j} & \quad \defeq \quad &  \anotfgbj_{s_0,q_0} \wedge (\anotfgnj_{s_0,q_0} = b) & \qquad & \text{with weight }1,\\
\soft_{\tiny\LTLfinally\LTLglobally}^{j} & \quad \defeq \quad & \anotfgbj_{s_0,q_0} & \quad & \text{with weight }n, \text{ and }\\
\soft_{\tiny\LTLglobally\LTLfinally}^{j} & \quad \defeq \quad & \anotfgbj_{s_0,q_0} \vee \anotbj_{s_0,\widehat  q_0} & \quad & \text{with weight }n^2.\\
\end{array}
\]

The definition of the soft constraints guarantees that $\trans \models \gphi_j$ if and only if there exist corresponding annotations that satisfy all three of the soft constraints for $\gphi_j$. Similarly, if $\trans \models \fgphi_j$, then $\soft_{\tiny\LTLfinally\LTLglobally}^{j}$ and $\soft_{\tiny\LTLglobally\LTLfinally}^{j}$ can be satisfied.
The weights of the soft clauses reflect the ordering of transition systems with respect to their satisfaction of $\gphi_1 \wedge \ldots \wedge \gphi_n$, as stated below. 
\begin{lemma}\label{lem:encoding-weights}
Let $\trans'$ and $\trans''$ be two transition systems such that $\trans' \models \spec$ and $\trans'' \models \spec$. Let $a'$ and $a''$ be the variable assignments satisfying the constraint system, such that $a'$ is an optimal assignment consistent with $\trans'$, and $a''$ is an optimal assignment consistent with $\trans''$. Furthermore, let $w'$ and $w''$ be the sums of the weights of the soft clauses satisfied by $a'$ and $a''$, respectively. Then, it holds that
$\val{\trans'}{\gphi_1 \wedge \ldots \wedge \gphi_n} < \val{\trans''}{\gphi_1 \wedge \ldots \wedge \gphi_n} \text{ iff } w' < w''.$
\end{lemma}
\begin{proof}
Let $(v_1',v_2',v_3')  = \val{\trans'}{\gphi_1 \wedge \ldots \wedge \gphi_n}$ and 
$(v_1'',v_2'',v_3'')  = \val{\trans''}{\gphi_1 \wedge \ldots \wedge \gphi_n}$. This means that there are exactly $v_3'$ distinct indices $i \in \{1,\ldots,n\}$ such that $\trans' \models \gphi_i$,  $v_2'$ distinct indices $i \in \{1,\ldots,n\}$ such that $\trans' \models \fgphi_i$ and $v_1'$ distinct indices $i \in \{1,\ldots,n\}$ such that $\trans' \models \gfphi_i$. Since $a'$ is an optimal satisfying assignment corresponding to $\trans'$, we have that $a'$ satisfies exactly $v_3'$ of the soft clauses $\soft_{\tiny\LTLglobally}^{j}$, exactly $v_2'$ of the soft clauses $\soft_{\tiny\LTLfinally\LTLglobally}^{j}$ and exactly $v_1'$ of the soft clauses $\soft_{\tiny\LTLglobally\LTLfinally}^{j}$. This means that $w' = v_3' + v_2' \cdot n + v_1' \cdot n^2$. In a similar way we can conclude that $w'' = v_3'' + v_2'' \cdot n + v_1'' \cdot n^2$ holds for $\trans''$.

First, suppose that $(v_1',v_2',v_3') < (v_1'',v_2'',v_3'')$. There are three possible cases:

\noindent
\emph{Case 1:} $v_1' = v_1''$, $v_2' = v_2''$ and $v_3' < v_3''$. Then, $w'' - w' = (v_3'' - v_3') > 0$.

\noindent
\emph{Case 2:} $v_1' = v_1''$ and $v_2' < v_2''$. Then, $w'' - w' = (v_2'' - v_2')\cdot n  + (v_3'' - v_3')$.
Since $\trans' \models \gphi_i$ implies $\trans' \models \fgphi_i$, we have that $v_3' - v_3'' \leq n-1$, due to the fact that $v_2'' - v_2' \geq 1$. Thus, we conclude $w'' - w' \geq n - (n-1) = 1 >0$.

\noindent
\emph{Case 3:} $v_1' < v_1''$. Now, $w'' - w' = (v_1'' - v_1')\cdot n^2 + (v_2'' - v_2')\cdot n  + (v_3'' - v_3')$. Again, since $\trans' \models \gphi_i$ implies $\trans' \models \gfphi_i$ and 
$\trans' \models \fgphi_i$ implies $\trans' \models \gfphi_i$, we have that $v_3' - v_3'' \leq n-1$ and 
$v_2' - v_2'' \leq n-1$, both due to the fact that $v_1'' - v_1' \geq 1$. Thus, we conclude $w'' - w' \geq n^2 - (n-1)\cdot n -(n-1) = 1 >0$.

In all three cases we showed that $w' < w''$.

For the other direction, suppose that $w' < w''$. If we assume that $(v_1',v_2',v_3') \geq (v_1'',v_2'',v_3'')$, then we can show as above that $w'' \geq w'$, which is a contradiction. Hence, we have that $(v_1',v_2',v_3') < (v_1'',v_2'',v_3'')$, which concludes the proof.\qed
\end{proof}

This in turn guarantees that  a transition system extracted from an optimal satisfying assignment for the MaxSAT problem is optimal with respect to the value of $\gphi_1 \wedge \ldots \wedge \gphi_n$, as stated in the following theorem that establishes the correctness of the encoding.
\begin{theorem}\label{thm:encoding-correctness}
Let $\autA$ be a given co-B\"uchi automaton for $\varphi$, and for each $j \in \{1,\ldots,n\}$, let $\autB_j = \relaxfg(\gphi_j)$ be the universal automaton for $\gphi_j$ constructed as in Section~\ref{sec:automata-safety}, and let $\autA_j$ be a universal co-B\"uchi automaton for $\gfphi_j$. 
The constraint system for bound $b \in \nats_{>0}$ is satisfiable if and only if there exists an implementation $\trans$ with $|\trans| \leq b$  such that $\trans \models \varphi$. Furthermore, from the optimal satisfying assignment to the variables $\tau_{s,\inpval,s'}$ and $o_{s,\inpval}$, one can extract a transition system $\trans^*$ such that for every transition system $\trans$ with $|\trans| \leq b$ and $\trans \models \varphi$ it holds that $\val{\trans^*}{\gphi_1 \wedge \ldots \wedge \gphi_n} \geq \val{\trans}{\gphi_1 \wedge \ldots \wedge \gphi_n}$.
\end{theorem}
\begin{proof}
	The first part of the claim follows from the correctness of the classical bounded synthesis approach. More precisely, if the constraint system is satisfiable, then there exists a satisfying assignment, which in particular, satisfies the constraints asserting the existence of a transition system $\trans$ of size less than or equal to $b$, and the existence of a valid annotation for the run graph of $\autA$ on $\trans$. If, on the other hand, there exists a transition system $\trans$ such that $|\trans| \leq b$ and $\trans \models \varphi$, then there exists a variable assignment $a$ consistent with $\trans$ that satisfies the constraints asserting the existence of a valid annotation for the run graph of $\autA$ on $\trans$. It remains to show that $a$ can be chosen in a way that satisfies the remaining hard constraints as well. To see that, notice that all the constraints for the annotations $\anotfg_j$ and $\anot_j$ can be satisfied (not necessarily in an optimal way) by setting all the variables $\anotfgbj_{s,q}$ and $\anotbj_{s,q}$ to $\falseval$. This completes the proof of the first statement.
	
	Now, let $a^*$ be an optimal solution to the MaxSAT problem, and $\trans^*$ be the transition system extracted from $a^*$. Consider a transition system $\trans$ such that $|\trans| \leq b$ and $\trans \models \varphi$. Then, as we showed above, there exists a satisfying assignment $a$ consistent with $\trans$. Let $w^*$ be the sum of the weights of the soft clauses satisfied by $a^*$, and $w$ be the sum of the weighs of the soft clauses satisfied by $a$. Since $a^*$ is an optimal satisfying assignment, we have that $w \leq w^*$. Thus, by applying Lemma~\ref{lem:encoding-weights} we obtain $\val{\trans^*}{\gphi_1 \wedge \ldots \wedge \gphi_n} \geq \val{\trans}{\gphi_1 \wedge \ldots \wedge \gphi_n}$, which concludes the proof of the second claim of the theorem.\qed
\end{proof}

\begin{wrapfigure}{r!}{0.3\textwidth}
	\centering
	\vspace{-0.9cm}
	\scalebox{.7}{
	\begin{tikzpicture}[node distance=1.5 cm,auto,>=latex',line join=bevel,transform shape]
\node[circle,draw] at (0,0) (s0) {$s_0$};
\node  [left of=s0,xshift=.5cm] (ini) {};
\node  [below left of =s0,yshift=-1cm,circle,draw] (s1) {$s_1$};
\node  [below right of =s0,yshift=-1cm,circle,draw] (s2) {$s_2$};
\draw [->] (s0) edge[loop above] node[right] {$\neg \mathit{r1}\wedge\neg \mathit{r2}$}(s0);
\draw [->] (s0) edge[bend right] node[above,sloped] {$\mathit{r1}\wedge\neg \mathit{r2}$} (s1);
\draw [->] (s0) edge node[left,near end] {$\mathit{r2}$} (s2);
\draw [->] (s1) edge node[below,sloped] {$\tiny\neg \mathit{r1}\wedge\neg \mathit{r2}$}(s0);
\draw [->] (s1) edge[loop left] node[below,yshift=-.3cm] {$\tiny\mathit{r1}\wedge\neg \mathit{r2}$}(s1);
\draw [->] (s1) edge node[above,midway] {$\mathit{r2}$} (s2);	
\draw [->] (s2) edge[bend right] node[right] {$\neg\mathit{r1}$}(s0);
\draw [->] (s2) edge[bend left] node[below] {$\mathit{r1}$} (s1);
\draw [->] (ini) edge (s0);
\end{tikzpicture}}
	\caption{An optimal implementation for Example~\ref{ex:two-safety}. }
	\label{fig:waiter-strategy}
\end{wrapfigure}
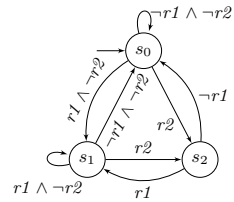\interlinepenalty 10000
Figure~\ref{fig:waiter-strategy} shows a transition system extracted from an optimal satisfying assignment for Example~\ref{ex:two-safety} with bound $3$ on the implementation size. The  transitions depicted in the figure are defined by the values of the variables $\tau_{s,\inpval,s'}$. The outputs of the implementation (omitted from the figure) are defined by the values of $o_{s,\inpval}$. The output in state $s_1$ when $\mathit{r1}$ is $\trueval$ is $\mathit{table1}\wedge \neg\mathit{table2}$, and the output in $s_2$ when $\mathit{r2}$ is $\trueval$ is $\neg\mathit{table1}\wedge \mathit{table2}$. For all other combinations of state and input, the output is $\neg\mathit{table1}\wedge \neg\mathit{table2}$.
\interlinepenalty 0

\clearpage
The next proposition establishes the size of the MaxSAT encoding.
\begin{proposition}\label{prop:encoding-size}
	Let $\autA$ be a given co-B\"uchi automaton for $\varphi$, and for each $j \in \{1,\ldots,n\}$, let $\autB_j = \relaxfg(\gphi_j)$ be the universal B\"uchi automaton for $\gphi_j$ constructed as in Section~\ref{sec:automata-safety}, and let $\autA_j$ be a universal co-B\"uchi automaton for $\gfphi$. 
	The constraint system for bound $b \in \nats$ has weights in $\mathcal{O}(n^2)$. It has
	\[
	\begin{array}{l}
	\mathcal{O}\Big(
	(b^2 + b \cdot |\outv|)\cdot 2^{|\inpv|} +  
	b \cdot |Q|\cdot (1+ \log(b \cdot |Q|)) +\\ 
	\phantom{\mathcal{O}(}\sum_{j=1}^n\big(b \cdot |Q_j| (1 + \log(b\cdot |Q_j|))\big)
	+ \sum_{j=1}^n\big(b \cdot |\widehat Q_j| (1 + \log(b\cdot |\widehat Q_j|))\big)
	\Big)
	\end{array}\]
	variables, and its size (before conversion to CNF) is
	\[
	\begin{array}{l}
	\mathcal{O}\Big(
	|Q|^2 \cdot b^2 \cdot 2^{|\inpv|} \cdot (d + \log(b\cdot |Q|)) + \\
	\phantom{\mathcal{O}(}\sum_{j=1}^n\big(|Q_j|^2 \cdot b^2 \cdot 2^{|\inpv|} \cdot (d_j + r_j + \log(b\cdot |Q_j|))\big) +
	\\
	\phantom{\mathcal{O}(}\sum_{j=1}^n\big(|\widehat Q_j|^2 \cdot b^2 \cdot 2^{|\inpv|} \cdot (\widehat d_j + \log(b\cdot |\widehat Q_j|))\big)
	\Big), 
	\end{array}\]
	\[
	\begin{array}{lllllll}
	\text{ where} &d &=& \max_{s,q,\inpval,q'}|\delta_{s,q,\inpval,q'}|,&
	d_j &=& \max_{s,q,\inpval,q'}|\delta_{s,q,\inpval,q'}^j|,\\
	&\widehat d_j &=& \max_{s,q,\inpval,q'}|\widehat \delta_{s,q,\inpval,q'}^j|, \text{ and } &
	r_j &=& \max_{s,q,\inpval,q'}|\rej^j(s,q,q',\inpval)|.
	\end{array}
	\]
\end{proposition}
\begin{proof}
The constraint system is defined in terms of the following variables:
\begin{itemize}
\item Boolean variables $\tau_{s,\inpval,s'}$ and $o_{s,\inpval}$ representing the transition system. The total number of these variables is $b^2 \cdot 2^{|\inpv|} + b \cdot |\outv| \cdot 2^{|\inpv|}$.
\item Boolean variables $\anotb_{s,q}$ and vectors of Boolean variables $\anotn_{s,q}$ representing the annotation $\anot$. The total number of bits is $b \cdot |Q| +b \cdot |Q|\cdot \log(b \cdot |Q|)$.
\item Boolean variables $\anotfgbj_{s,q}$ and vectors of Boolean variables $\anotfgnj_{s,q}$ representing the annotations $\anotfg_j$. The total number of bits is $\sum_{j=1}^n\big(b \cdot |Q_j| (1 + \log(b\cdot |Q_j|))\big)$.
\item Boolean variables $\anotbj_{s,q}$ and vectors of Boolean variables $\anotnj_{s,q}$ representing the annotation $\anot_j$. The total number of bits is $\sum_{j=1}^n\big(b \cdot |\widehat Q_j| (1 + \log(b\cdot |\widehat Q_j|))\big)$. 
\end{itemize}

The sum of the above quantities yields the total number of Boolean variables.

The constraint system consists of the following constraints:
\begin{itemize}
\item Constraints $C_\tau$ encoding input-enabledness, of size $b^2 \cdot 2^{|\inpv|}$.
\item Constraints $C_\anot$ for valid annotation $\anot$ of size $\mathcal{O}\big(|Q|^2 \cdot b^2 \cdot 2^{|\inpv|} \cdot (d + \log(b\cdot |Q|))\big)$.
\item Hard constraints $C_\anotfg^{j}$ for valid annotations $\anotfg_{j}$, of size \[\mathcal{O}\big(\sum_{j=1}^n\big(|Q_j|^2 \cdot b^2 \cdot 2^{|\inpv|} \cdot (d_j + r_j + \log(b\cdot |Q_j|))\big)\big).\]
\item Hard constraints $C_\anot^{j}$ for valid annotations $\anot_{j}$, of size \[\mathcal{O}\big(\sum_{j=1}^n\big(|\widehat Q_j|^2 \cdot b^2 \cdot 2^{|\inpv|} \cdot (\widehat d_j + \log(b\cdot |Q_j|))\big)\big).\]
\item Soft constraints $\soft_{\tiny\LTLglobally}^{j}$, $\soft_{\tiny\LTLfinally\LTLglobally}^{j}$ and $\soft_{\tiny\LTLglobally\LTLfinally}^{j}$ for valid annotations, of size
\[\mathcal{O}\big(\sum_{j=1}^n\big(|\log(b\cdot |Q_j|)\big)\big).\]
\end{itemize}

Summing up, we obtain the total size of the constraint system.\qed
\end{proof}

\subsection{Generalizations of the Maximum Realizability Problem}\label{sec:generalizations}
	
\subsubsection{Maximum Realizability with Soft LTL Specifications}\label{sec:soft-ltl}

The first generalization of the maximum realizability problem that we consider is the setting where the soft specifications can be arbitrary LTL formulas, and not just safety properties of specific form. More precisely, each soft specification $\softSpec$ is an LTL formula for which we are also given a vector $\Relax(\softSpec)$ of LTL formulas that defines the possible relaxations of $\softSpec$. Formally, $\Relax(\softSpec) = (\psi_{1},\ldots,\psi_{m})$, where $\psi_{1} = \softSpec$, and for every $1 \leq k < m$, it holds that $\trans \models \psi_{k}$ implies $\trans \models \psi_{k+1}$ for every transition system $\trans$. That is, $\psi_{1},\ldots,\psi_{m}$ are ordered according to strength.
In particular, if $\trans \models \softSpec$, then $\trans \models \psi_k$ for each $\psi_k$ in $\Relax(\softSpec)$.
For example, if $\softSpec = \LTLglobally p$ for some atomic proposition $p$, we can take $\Relax(\softSpec) = (\LTLglobally p,\LTLfinally\LTLglobally p,\LTLglobally\LTLfinally p)$.

As in Section~\ref{sec:quantitative-semantics}, we define the value of $\softSpec$ for given $\Relax(\softSpec) = (\psi_{1},\ldots,\psi_{m})$ to be $\val\trans\softSpec = (v_1,\ldots,v_m)$, where $v_k = 1$ if $\trans \models \psi_{(m+1)-k}$, and $v_k = 0$ otherwise.
For a conjunction $\softSpec_1\wedge\ldots\wedge\softSpec_n$ of soft specifications with given $\Relax(\softSpec_j) = (\psi_{j,1},\ldots,\psi_{j,m})$ for each $j \in \{1,\ldots,n\}$, we define the value
$\val\trans{\softSpec_1 \wedge \ldots \wedge \softSpec_n} = \big(
\sum_{i=1}^n v_{i,1},
\ldots,
\sum_{i=1}^n v_{i,m}
\big),$ where
$\val\trans{\softSpec_i} = (v_{i,1},\ldots,v_{i,m})$ for $i \in \{1,\ldots,n\}$.

The maximum realizability problem asks for a given LTL specification $\spec$ and soft LTL specifications 
$\softSpec_1,\ldots,\softSpec_n$ with given $\Relax(\softSpec_j) = (\psi_{j,1},\ldots,\psi_{j,m})$, to determine whether there exists a transition system $\trans$ such that $\trans\models\spec$, and if the answer is positive, to construct a transition system $\trans^*$ such that $\trans^* \models \spec$, and for every $\trans$ with $\trans\models\spec$, it holds that $\val\trans{\softSpec_1 \wedge \ldots \wedge \softSpec_n} \leq \val{\trans^*}{\softSpec_1 \wedge \ldots \wedge \softSpec_n}$.
The bounded maximum realizability problem is defined in the straightforward way.

We can adapt the MaxSAT approach from Section~\ref{sec:maxsat} to solve the bounded maximum realizability problem in this setting as follows. 

First, for each $\psi_{j,k}$ in $\Relax(\softSpec_j)$, we construct a universal co-B\"uchi automaton $\autA_{j,k} = (Q_{j,k},q_0^{j,k},\delta_{j,k},F_{j,k})$ such that $\trans \in \lang{\autA_{j,k}}$ if and only if $\trans \models \psi_{j,k}$.
In the MaxSAT encoding, the hard constraints for the annotation $\anot_{j,k}$  are
\begin{align*}
C_{j,k}\defeq\bigwedge_{q,q' \in Q_{j,k}}\bigwedge_{s,s' \in S}\bigwedge_{\inpval \in \ialphabet}
\Bigg( &
\big(
\anotbjk_{s,q} \wedge
\delta^{j,k}_{s,q,\inpval,q'}\wedge 
\tau_{s,\inpval,s'}
\big) \rightarrow 
\succa_{\anot}^{j,k}(s,q,s',q',\inpval)
\Bigg).
\end{align*}
Generalizing the encoding, for each $j \in \{1,\ldots,n\}$ and $k \in \{1,\ldots,m\}$, we now have one soft constraint
$\soft_{j,k}\defeq  \bigvee_{l=1}^k \anotbjl_{s_0,q^{j,l}}$ with weight $n^{k-1}$.

\subsubsection{Maximum Realizability with Priorities}\label{sec:priorities}
	
In the definitions in Section~\ref{sec:prob-form} and the paragraph above, all soft specifications have the same priority. Now, we extend the maximum realizability setting to the case with priorities for the soft specifications given as part of the input to the problem.

\smallskip
\noindent
{\it Soft specifications with priority ordering.} We begin with a simple setting where soft specifications are simply ordered in decreasing priority, without assigning any numerical weight for the preferences over the formulas. More specifically, we assume that the soft specifications $\softSpec_1,\ldots,\softSpec_n$ are ordered such that, for every $i \in \{1,\ldots,n\}$, we have that $\softSpec_i$ has higher priority than $\softSpec_j$ for all $j > i$. 

Now, given a vector $\Relax(\softSpec_j) = (\psi_{j,1},\ldots,\psi_{j,m})$ for $\softSpec_j$  we define the value of $\softSpec_j$ in a transition system $\trans$ to be the number of specifications in $\Relax(\softSpec_j)$ satisfied by $\trans$, i.e., $\val\trans{\softSpec_j} = | \{k \in \{1,\ldots,m\} \mid \trans\models \psi_{j,k}\}|$. The value of $\softSpec_1\wedge\ldots\wedge\softSpec_n$ is then defined as $\val\trans{\softSpec_1 \wedge \ldots \wedge \softSpec_n} = (\val\trans{\softSpec_1},\ldots,\val\trans{\softSpec_n})$. The values of transition systems are compared according to the lexicographic ordering of vectors in $\{0,\ldots,m\}^n$, thus giving priority to $\softSpec_i$ over $\softSpec_j$ for $i < j$.

The MaxSAT approach can be adapted for this value function in the same way as above. The difference is in the weights of the soft constraints for the annotations $\anot_{j,k}$: for each
$j \in \{1,\ldots, n\}$ and $k \in \{1,\ldots,m\}$ we have a soft constraint  $\soft_{j,k} \; \defeq \;  \bigvee_{l=1}^k \anotbjl_{s_0,q^{j,l}}$ with weight $w_{j,k}$, where 
$w_{j,k} = 1 $ if $j=n$ or $k < m$, and $w_{j,k} = \sum_{j' = j+1}^n \sum_{k = 1}^m w_{j',k} + 1$ otherwise.

\smallskip
\noindent
{\it Soft specifications with given weights.}
We also consider the weighted maximum realizability problem, in which, together with $\Relax(\softSpec)$ for each soft specification $\softSpec$, the user also provides numerical weights for the formulas in $\Relax(\softSpec)$. That is, for each $j \in \{1,\ldots,n\}$ and $k \in \{1,\ldots,m\}$, we are given a weight $w_{j,k}$ for $\psi_{j,k}$. These weights specify the priority of each of the soft specifications.

The MaxSAT formulation is then adapted to incorporate the given weights, by using them for the corresponding soft constraints. Namely, for each $j$ and $k$, the corresponding soft constraint $\soft_{j,k}  \; \defeq \;  \bigvee_{l=1}^k \anotbjl_{s_0,q^{j,l}}$ has weight $w_{j,k}$.

\begin{theorem}\label{thm:optimal-bound-general}
	Given an LTL specification $\spec$ and soft specifications $\gphi_1,\ldots,\gphi_n$ together with a vector of formulas $\Relax(\softSpec_j) = (\psi_{j,1},\ldots,\psi_{j,m})$ for each $\softSpec_j$,
	if there is a transition system $\trans$ with  $\trans \models \spec$, then there exists $\trans^*$ such that:
	\begin{itemize}
		\item $\val\trans{\softSpec_1 \wedge \ldots \wedge \softSpec_n} \leq \val{\trans^*}{\softSpec_1 \wedge \ldots \wedge \softSpec_n}$ for all $\trans$ with $\trans \models \spec$, and
		\item $\trans^* \models \spec$ and $|\trans^*| \leq (2^{b+\log b})!^2$,
	\end{itemize}
	where $b = \max\{|\subf{\spec\wedge\softSpec_1'\wedge\ldots\wedge\softSpec_n'}| \mid \softSpec_i' \in \Relax(\softSpec_i) \text{ for }i=1,\ldots,n\}$.
\end{theorem}
\begin{proof}
	The proof is a generalization of the proof of Theorem~\ref{thm:optimal-bound-safety}. First, we need to establish the analogue of Lemma~\ref{lem:value-as-ltl} for the general case. 
	
	\begin{lemma}\label{lem:value-as-ltl-general}
		For every transition system $\trans$, soft specifications $\gphi_1,\ldots,\gphi_n$, and vector of formulas $\Relax(\softSpec_j) = (\psi_{j,1},\ldots,\psi_{j,m})$ for each $\softSpec_j$, if $\val\trans{\gphi_1 \wedge \ldots \wedge \gphi_n} = v$, then there exists an LTL formula $\psi_v$ such that $\trans \models \psi_v$ and the following holds:
		\begin{itemize}
			\item[(1)] $\psi_v = \softSpec_1'\wedge\ldots\wedge\softSpec_n'$, where $\softSpec_i' \in\Relax(\softSpec_i) \cup \{\mathit{true}\} \text{ for }i=1,\ldots,n$, 
			\item[(2)] for every $\trans'$, if $\trans' \models \psi_v$, then $\val{\trans'}{\gphi_1 \wedge \ldots \wedge \gphi_n} \geq v$.
		\end{itemize}
	\end{lemma}
	The proof of Lemma~\ref{lem:value-as-ltl-general} is analogous to the proof of Lemma~\ref{lem:value-as-ltl}. Then, with the help of Lemma~\ref{lem:value-as-ltl-general} we can establish Theorem~\ref{thm:optimal-bound-general} in the same way as Theorem~\ref{thm:optimal-bound-safety}. \qed
\end{proof}

\section{Experimental Evaluation}\label{sec:experiments}
We implemented the proposed approach to maximum realizability\footnote{The code is available at https://github.com/MahsaGhasemi/max-realizability}  in Python 2.7.  For the LTL to automata translation Spot~\cite{Duret-LutzLFMRX16} version 2.2.4 is used. MaxSAT instances are solved by Open-WBO~\cite{MartinsML14} version 2.0. We evaluated our method on instances of two examples. Each experiment was run on a machine with a 2.3 GHz Intel Xeon E5-2686 v4 processor and 16 GiB of memory. While the processor is quad-core, only a single core was used. We set a time-out of 1 hour.

\subsection{Robotic Navigation.}
We applied our method to the strategy synthesis for a robotic museum guide. The map of the museum is shown in Figure~\ref{fig:map}. The robot has to give a tour of the exhibitions in a specific order, which constitutes the hard specification. The tour starts at the entrance of the museum where the robot picks up newly arrived visitors. The main objective is to take the group through the two exhibitions on that floor and then return to the entrance to pick up a new group of people. 
Preferably, it also avoids certain locations, such as the library, or the passage when it is occupied. These preferences are encoded in the soft specifications.
In particular, on one hand, the robot can only gain access to Exhibition 2 by getting a key from the staff's office. On the other hand, the robot is asked not to disturb the employees in the office. There is a library between Exhibition 1 and Exhibition 2 which can be used to go from one to the other, but it is preferred that visitors do not enter the library. However, it is also desirable that when the other passage between these two exhibitions is occupied, the robot does not go through there. 
 
\begin{figure*}
	\centering
	\includegraphics[width=0.8\linewidth]{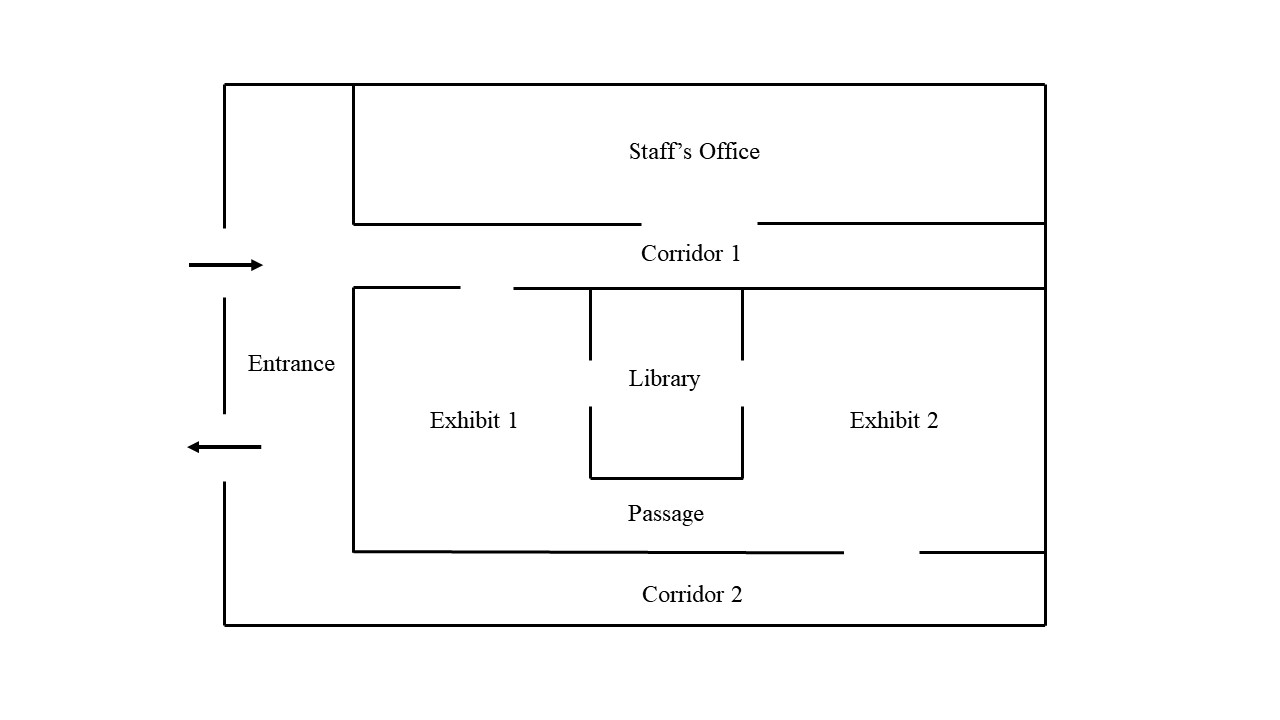}
	\caption{Map of the museum.}
	\label{fig:map}
\end{figure*} 

Clearly, these specifications cannot be realized in conjunction. Given their priorities, we categorize the requirements into hard and soft specifications, and synthesize a strategy which satisfies the hard specifications and maximizes the satisfaction of the soft specifications. We formalize the problem as follows.

\smallskip
\noindent
{\bf Input propositions:} 
The set $\inpv$ contains a single Boolean variable $\occupied$ that indicates whether the passage between the two exhibitions is occupied.

\smallskip
\noindent
{\bf Output propositions:}
The set of output propositions $\outv$ consists of eight Boolean variables corresponding to the eight locations on the map: $\ent$, $\corr_1$, $\corr_2$, $\exh_1$, $\exh_2$, $\passage$, $\office$, $\library$.

\smallskip
\noindent
{\bf The hard specification} is the conjunction of the following formulas.

\begin{itemize}
	\item The robot starts at the entrance:
	\begin{equation*}
	\ent.
	\end{equation*}
	
	\item At each time step, the robot can occupy only one location:
	\begin{equation*}
	\LTLglobally \bigwedge_{o_1 \in \outv} \left( o_1 \rightarrow \bigwedge_{o_2 \in \outv \backslash \{o_1\}} \neg o_2 \right) .
	\end{equation*}
	
	\item The admissible actions of the robot are to stay in the current location or move to an adjacent one. This leads to eight requirements describing the map. 
	For instance:
	\begin{equation*}
	\LTLglobally \left( \corr_1 \rightarrow \LTLnext \left(\corr_1 \lor \office \lor \exh_1 \right) \right) .
	\end{equation*}
\end{itemize}

{\it Remark: }Due to the requirements above,  the robot will always be in exactly one valid location, i.e., in a transition system that satisfies the specifications it is impossible to reach a state where all output variables are false.

\begin{itemize}
	
	\item The robot must infinitely often visit both exhibitions:
	\begin{equation*}
	\begin{aligned}
	& \LTLglobally \LTLfinally \exh_1 ,\\
	& \LTLglobally \LTLfinally \exh_2 .
	\end{aligned}
	\end{equation*}
	
	\item The robot has to respect the order of visits, by starting from Exhibition 1, going to Exhibition 2 and finishing at the entrance:
	\begin{equation*}
	\begin{aligned}
	& \LTLglobally \left(\exh_1 \rightarrow \LTLnext \left( \left( \neg \ent \land \neg \exh_1 \right) \LTLuntil \exh_2 \right) \right) ,\\
	& \LTLglobally \left(\exh_2 \rightarrow \LTLnext \left( \left( \neg \exh_1 \land \neg \exh_2 \right) \LTLuntil \ent \right) \right) ,\\
	& \LTLglobally \left(\ent \rightarrow \LTLnext \left( \left( \neg \exh_2 \land \neg \ent \right) \LTLuntil \exh_1 \right) \right) .
	\end{aligned}
	\end{equation*}
	
	\item The robot does not have access to Exhibition 2 before it visits the office:
	\begin{equation*}
	\neg \exh_2 \LTLuntil\office.
	\end{equation*}
	
\end{itemize}

\smallskip
\noindent
{\bf The set of soft specifications} describes the desirable requirements that the robot does not enter the office, the library, or a occupied passage. Formally:

\begin{itemize}
	\item The robot must not enter the office from corridor 1:
	\begin{equation*}
	\LTLglobally \left(\corr_1 \rightarrow \LTLnext \neg \office \right) .
	\end{equation*}
	
	\item The robot must not enter the library from the exhibitions:
	\begin{equation*}
	\LTLglobally \left( \exh_1 \lor \exh_2 \rightarrow \LTLnext \neg \library \right) .
	\end{equation*}
	
	\item The robot must not enter the passage from the exhibitions when it is occupied:
	\begin{equation*}
	\LTLglobally \left( \left( \exh_1 \lor \exh_2 \right) \land \LTLnext \occupied \rightarrow \LTLnext \neg \passage \right) .
	\end{equation*}
	
\end{itemize}

We applied the proposed method described in Section~\ref{sec:maxsat-encoding} on this example.
Table~\ref{tab:inst-robot} summarizes the results. With implementation bound of 8, the hard specification is realizable and a partial satisfaction of soft specifications is achieved. This strategy always selects the passage to transition from Exhibition 1 to Exhibition 2 and hence, avoids the library. It also violates the requirement of not entering the staff's office, to acquire access to Exhibition 2. For implementation bound 10 the solver times out. Notice that strategies with higher values exists, however, they require larger implementation size.

\begin{table*}[h!]\centering\footnotesize
	\caption{Results of applying synthesis with maximum realizability to the robotic navigation example, with different bounds on implementation size $|\trans|$. We report on the number of variables and clauses in the encoding, the satisfiability of hard constraints, the value (and bound) of the MaxSAT objective function, the running times of Spot, Open-WBO, and the time of the solver plus the time for generating the encoding.}
	\begin{tabular*}{1\linewidth}{ccccccccccc}
		\toprule
		&& \multicolumn{2}{c}{Encoding} && \multicolumn{2}{c}{Solution} && \multicolumn{3}{c}{Time (s)} \\
		\cline{3-4} \cline{6-7} \cline{9-11}
		$|\trans|$ & \phantom{} & \# vars & \# clauses & \phantom{} & sat. & $\Sigma weights$ & \phantom{} & Spot & Open-WBO & enc.+solve \\
		\midrule
		2 && 4051 & 25366 && UNSAT & 0 (39) && 0.93 & 0.011 & 0.12 \\
		4 && 19965 & 125224 && UNSAT & 0 (39) && 0.93 & 0.079 & 0.57 \\
		6 && 45897 & 289798 && UNSAT & 0 (39) && 0.93 & 1.75 & 2.9 \\
		8 && 95617 & 596430 && SAT & 31 (39) && 0.93 & 956 & 959 \\
		10 && 152949 & 954532 && SAT & - (39) && 0.93 & time-out & time-out \\
		\bottomrule
	\end{tabular*}
	\label{tab:inst-robot}
\end{table*}

\subsection{Power Distribution Network.}\looseness=-1

We consider the problem of dynamic reconfiguration of power distribution networks. A power network consists of a set $\supplies$ of power supplies (generators) and a set $\loads$ of loads (consumers). The network is a bipartite graph with edges between supplies and loads, where each supply is connected to multiple loads and each load is connected to multiple supplies. Each power supply has an associated capacity, which determines how many loads it can power at a given time. 
For each supply $p \in \supplies$, we denote with $\caps_p$ the capacity of $p$, that is, how many loads $p$ can power, and with $\cons(p)$ the set of loads connected to $p$ in the network graph. Similarly, for a load $l \in \loads$, let $\supl(l)$ be the set of suppliers to which the load $l$ is connected in the network.
It is possible that not all loads can be powered all the time. Some loads are critical and must be powered continuously (hard specification), while others are not and should be powered when possible (soft specification). Some loads can be initializing, meaning they must be powered only initially for several steps. Power supplies can become faulty during operation, which necessitates dynamic network reconfiguration. The number of supplies that can be simultaneously faulty is upper bounded by a constant $f$. Further, we add a soft specification to limit the frequency of switching the relays between power supplies and loads. We model the problem as follows.

\smallskip
\noindent
{\bf Input propositions:} The set $\inpv$ consists of input propositions which form the binary encoding of $f$ integer variables $e_1, e_2, \ldots, e_f$ each with domain $\{0, \ldots, |\loads|\}$. The values of these variables indicate which power supplies are faulty at a given point in time: if $e_i = p$ for some $i$ and $p$, then $p$ is faulty at that time.\footnote{There can be different indices $i$ for which $e_i = p$ at the same time. While this will be redundant, it does not affect the encoding.}

\smallskip
\noindent
{\bf Output propositions:}
The set of output propositions $\outv$ consists of the Boolean variables $\slp$ for all loads $l \in \loads$ and  supplies $p \in \supplies$ where $l$ and $p$ are connected. The meaning of $\slp$ being true is that $l$ is powered by $p$.

\smallskip
\noindent
{\bf The hard specification} is the conjunction of the following formulas.

\begin{itemize}
	\item A critical load must always be powered:
	\begin{equation*}
	\LTLglobally \left( \bigvee_{p \in \supl(l)} \slp \right) \,, 
	\quad \forall \ l \in \loads \ \text{where } l \text{ is critical}.
	\end{equation*}
	
	\item An initializing node must be powered during the first two steps:
	\begin{equation*}
	\left(\bigvee_{p \in \supl(l)} \slp\right) \land\,\,
	\LTLnext \left(\bigvee_{p \in \supl(l)} \slp\right)
	\,, \quad \forall \ l \in \loads \ \text{where } l \text{ is initializing}.
	\end{equation*}
	
	\item A load must only be assigned to one power supply:
	\begin{equation*}
	\bigwedge_{l \in \loads}\,\,\,
	\bigwedge_{p_1 \in \supl(l)}
	\LTLglobally \left(s_{l \rightarrow p_1} 
	\rightarrow \bigwedge_{p_2 \in \supl(l), p_2 \neq p_1}\neg s_{l \rightarrow p_2}\right).
	\end{equation*}
	
	\item The capacity of power supplies must not be exceeded:
	\begin{equation*}
	\bigwedge_{p \in \supplies} 
	\bigwedge_{\substack{{L' \subseteq \cons(p)} \\ {|L'| = \caps_p}}}
	\LTLglobally \left( \left(\bigwedge_{l \in L'} s_{l \rightarrow p}\right) \rightarrow
	\bigwedge_{l \in \cons(p) \setminus L'} \neg s_{l \rightarrow p} \right).
	\end{equation*}
	
	\item When a power supply becomes faulty, no loads can be powered by it:
	\begin{equation*}
	\bigwedge_{i \in \{1, 2, \ldots f\}}
	\bigwedge_{p \in \supplies}
	\LTLglobally \left( e_i = p \rightarrow 
	\bigwedge_{l \in \cons(p)} \neg \slp \right).
	\end{equation*}
	
\end{itemize}

\smallskip
\noindent
{\bf The set of soft specifications} consists of the requirements for powering the non-vital loads, and optionally, a restriction on switching supplies too often unless they become faulty. The respective formulas are given below.

\begin{itemize}
	\item A non-critical load should always be powered:\begin{equation*}
	\LTLglobally \left( \bigvee_{p \in \supl(l)} \slp \right) \,, 
	\quad \forall \ l \in \loads \ \text{where } l \text{ is non-critical}.
	\end{equation*}

	\item All loads powered by a supply remain powered by it unless it becomes faulty:
	\begin{equation*}
	\bigwedge_{l \in \loads}\,\,
	\bigwedge_{p \in \supl(l)}
	\LTLglobally \left( 
	\slp \land \LTLnext \left( \neg 
	\bigvee_{i \in \{1, 2, \ldots, f\}} e_i = l \right)
	\rightarrow \LTLnext \slp \right). 
	\end{equation*}
	
\end{itemize}

\begin{table*}[]\centering\footnotesize
\caption{Power distribution network instances. 
An instance is determined by the number supplies $|\supplies|$, the number of loads $|\loads|$, the capacity of supplies $\caps$, the number of critical, non-critical and initializing loads. We also show the number of input $|\inpv|$ and output $|\outv|$ propositions and the number of soft specifications.}
\begin{tabular*}{1\linewidth}{@{}c|cccccccccccc@{}}
\toprule
& &\multicolumn{3}{c}{Network} &&\multicolumn{3}{c}{Load characterization} && \multicolumn{3}{c}{Specifications} \\
\cline{3-5} \cline{7-9} \cline{11-13}
& \begin{tabular}{c@{}} Instance \\ \# \end{tabular} & $|\supplies|$ & $|\loads|$ & $\caps$ & \phantom{} & crit. & non-crit. & init. & \phantom{} & $|\inpv|$ & $|\outv|$ & \begin{tabular}{@{}c@{}} \# Soft \\ \; spec. \end{tabular}\\
\midrule
fully     & 1 & 3 & 3 &  1 && 1 & 2 & 0 &&  2 & 9 & 2 \\
connected, & 2 & 3 & 6 &  2 && 2 & 4 & 0 &&  2 & 18 & 4 \\
switching & 3 & 3 & 3 &  1 && 0 & 2 & 1 &&  2 & 9 & 2 \\
allowed & 4 & 3 & 6 &  2 && 1 & 4 & 1 &&  2 & 18 & 4 \\
\midrule
sparse,&5 & 4 & 2 &  1 && 1 & 1 & 0 &&  3 & 4 & 1 \\
switching&6 & 4 & 4 &  1 && 1 & 3 & 0 &&  3 & 8 & 3 \\
allowed&7 & 4 & 6 &  1 && 1 & 5 & 0 &&  3 & 12 & 5 \\
&8 & 4 & 8 &  1 && 1 & 7 & 0 &&  3 & 16 & 7 \\
\midrule
sparse,&9 & 4 & 2 &  1 && 1 & 1 & 0 &&  3 & 4 & 5 \\
switching&10 & 4 & 4 &  1 && 1 & 3 & 0 &&  3 & 8 & 11 \\
restricted&11 & 4 & 6 &  1 && 1 & 5 & 0 &&  3 & 12 & 17 \\
&12 & 4 & 8 &  1 && 1 & 7 & 0 &&  3 & 16 & 23 \\
\bottomrule
\end{tabular*}
\label{tab:def-inst}
\end{table*}

We applied our method to the problem of synthesizing a relay-switching strategy from the above LTL specifications. 
Table~\ref{tab:def-inst} describes the instances to which we applied our synthesis method. Power supplies have the same capacity $\caps$ (number of loads they can power) and at most one can be faulty at each time. We consider three categories of instances, depending on the network connectivity (full or sparse), and whether we restrict frequent switching of supplies.

\setlength\extrarowheight{-0.5pt}
\begin{table*}[h!]\centering\scriptsize
	\caption{Results of applying synthesis with maximum realizability on the instances in Table~\ref{tab:def-inst}, with different bounds on implementation size $|\trans|$. We report on the number of variables and clauses in the encoding, the value (and bound) of the objective function in the MaxSAT instance,  the running times of Spot, Open-WBO, and the time of the solver plus the time for generating the encoding.}
	\begin{tabular*}{1\linewidth}{@{}c|ccccccccccc@{}}
		\toprule
		&& && \multicolumn{2}{c}{Encoding} && \multicolumn{1}{c}{Solution} && \multicolumn{3}{c}{Time (s)} \\
		\cline{5-6} \cline{8-8} \cline{10-12}
		\begin{tabular}{@{}c@{}} Instance \\ \# \end{tabular} & \phantom{} & $|\trans|$ & \phantom{} & \# vars & \# clauses & \phantom{} & $\Sigma weights$ & \phantom{} & Spot & Open-WBO & enc.+solve \\
		\midrule
		1 && 2 && 246 & 3183 && 8 (14) && 0.23 & 0.0050 & 0.044 \\
		  && 4 && 1038 & 18059 && 8 (14) && 0.23 & 0.046 & 0.16 \\
		  && 6 && 2862 & 52999 && 8 (14) && 0.23 & 2.9 & 3.2 \\
		  && 8 && 4838 & 94079 && 8 (14) && 0.23 & 2305 & 2306 \\
		\midrule
		2 && 2 && 452 & 13429 && 64 (84) && 91 & 0.015 & 0.15 \\
		  && 4 && 1860 & 77125 && 69 (84) && 91 & 0.15 & 0.66 \\
		  && 6 && 5100 & 226933 && 69 (84) && 91 & 8.5 & 9.9 \\
		  && 8 && 8588 & 403165 && N/A (84) && 91 & time-out & time-out \\
		\midrule
		3 && 2 && 302 & 7567 && 8 (14) && 0.23 & 0.0067 & 0.069 \\
		  && 4 && 1446 & 42891 && 13 (14) && 0.23 & 0.021 & 0.25 \\
		  && 6 && 4206 & 125287 && 13 (14) && 0.23 & 0.11 & 0.73 \\
		  && 8 && 7206 & 222591 && 13 (14) && 0.23 & 1.1 & 2.5 \\
		\midrule
		4 && 2 && 508 & 38165 && 64 (84) && 210 & 0.028 & 0.34 \\
		  && 4 && 2268 & 219589 && 74 (84) && 210 & 0.21 & 1.5 \\
		  && 6 && 6444 & 645397 && 74 (84) && 210 & 2.7 & 6.4 \\
		  && 8 && 10956 & 1147101 && N/A (84) && 210 & time-out & time-out \\
		\midrule\midrule
		5 && 2 && 203 & 2476 && 3 (3) && 0.058 & 0.017 & 0.039 \\
		  && 4 && 779 & 14126 && 3 (3) && 0.058 & 0.24 & 0.31 \\
		  && 6 && 2019 & 41584 && 3 (3) && 0.058 & 1.4 & 1.6 \\
		  && 8 && 3395 & 73808 && 3 (3) && 0.058 & 4.0 & 4.3 \\
		\midrule
		6 && 2 && 433 & 6722 && 31 (39) && 0.17 & 0.0069 & 0.066 \\
		  && 4 && 1649 & 38472 && 31 (39) && 0.17 & 0.076 & 0.29 \\
		  && 6 && 4329 & 113422 && 31 (39) && 0.17 & 4.6 & 5.2 \\
		  && 8 && 7241 & 201326 && N/A (39) && 0.17 & time-out & time-out \\
		\midrule
		7 && 2 && 663 & 13464 && 100 (155) && 1.3 & 0.011 & 0.12 \\
		  && 4 && 2519 & 77538 && 106 (155) && 1.3 & 0.11 & 0.54 \\
		  && 6 && 6639 & 229036 && 106 (155) && 1.3 & 6.3 & 7.4 \\
		  && 8 && 11087 & 406668 && N/A (155) && 1.3 & time-out & time-out \\
		\midrule
		8 && 2 && 893 & 29070 && 196 (399) && 30 & 0.019 & 0.26 \\
		  && 4 && 3389 & 169596 && 294 (399) && 30 & 0.47 & 1.5 \\
		  && 6 && 8949 & 503338 && 294 (399) && 30 & 62 & 65 \\
		  && 8 && 14933 & 894122 && N/A (399) && 30 & time-out & time-out \\
		\midrule\midrule
		9 && 2 && 631 & 7352 && 131 (155) && 0.21 & 0.0069 & 0.10 \\
		  && 4 && 2647 & 43106 && 131 (155) && 0.21 & 0.15 & 0.44 \\
		  && 6 && 7311 & 129100 && 131 (155) && 0.21 & 71 & 71 \\
		  && 8 && 12367 & 229004 && N/A (155) && 0.21 & time-out & time-out \\
		\midrule
		10 && 2 && 1289 & 16474 && 1343 (1463) && 0.44 & 0.012 & 0.21 \\
		   && 4 && 5385 & 96432 && 1343 (1463) && 0.44 & 1.1 & 1.7 \\
		   && 6 && 14913 & 288454 && 1343 (1463) && 0.44 & 3579 & 3581 \\
		   && 8 && 25185 & 511718 && N/A (1463) && 0.44 & time-out & time-out \\
		\midrule
		11 && 2 && 1947 & 28092 && 4660 (5219) && 1.9 & 0.021 & 0.35 \\
		   && 4 && 8123 & 164478 && 4678 (5219) && 1.9 & 1.7 & 2.7 \\
		   && 6 && 22515 & 491584 && N/A (5219) && 1.9 & time-out & time-out \\
		   && 8 && 38003 & 872256 && N/A (5219) && 1.9 & time-out & time-out \\
		\midrule
		12 && 2 && 2605 & 48574 && 10724 (12719) && 28 & 0.056 & 0.61 \\
		   && 4 && 10861 & 285516 && 11686 (12719) && 28 & 1.7 & 3.5 \\
		   && 6 && 30117 & 853402 && N/A (12719) && 28 & time-out & time-out \\
		   && 8 && 50821 & 1514906 && N/A (12719) && 28 & time-out & time-out \\
		\bottomrule
	\end{tabular*}
	\label{tab:inst-stat}
\end{table*}

\begin{figure}[]
	\centering
	\begin{subfigure}{.7\textwidth}
		\centering
		\includegraphics[width=1\textwidth]{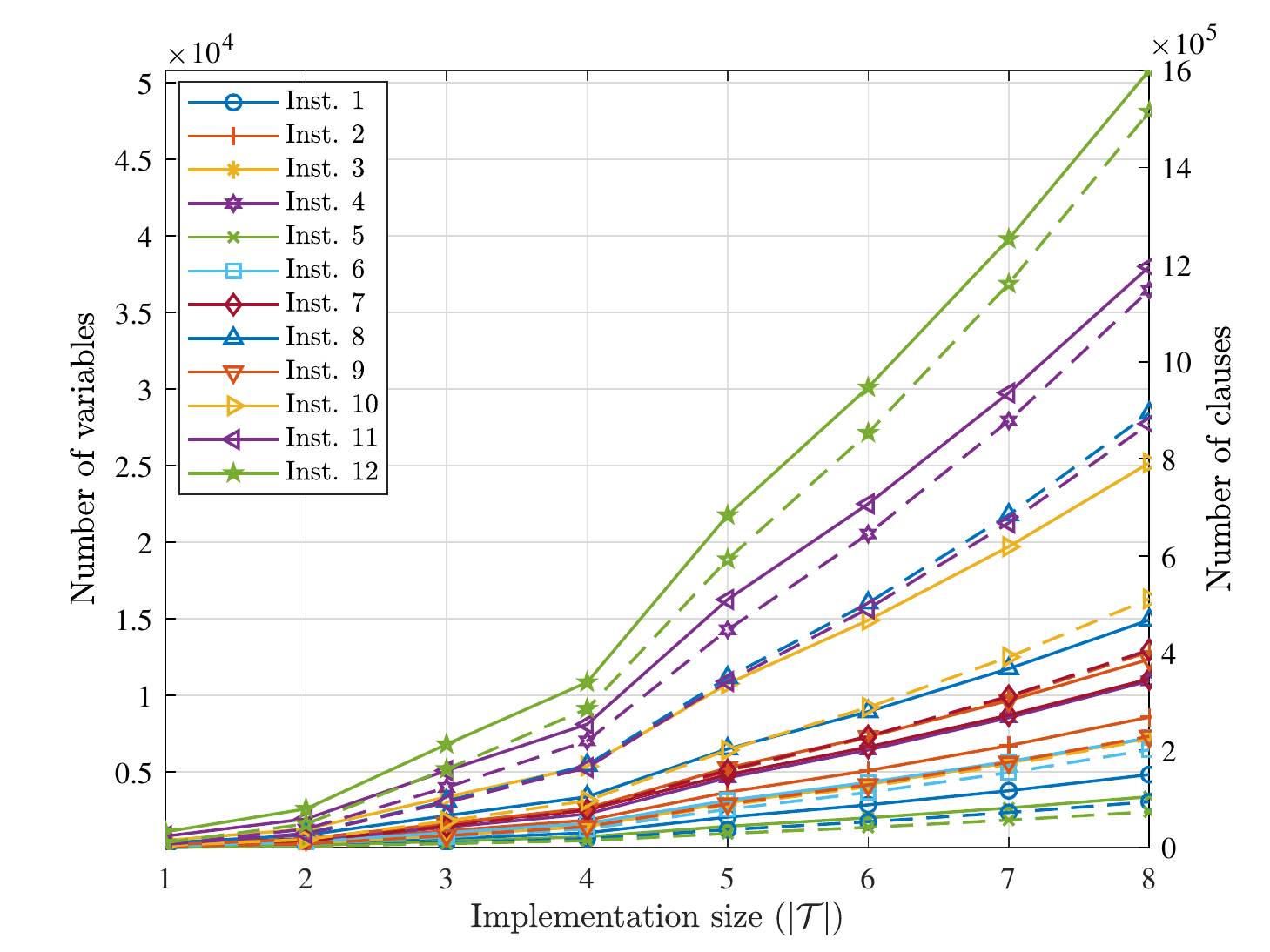}\caption{\scriptsize Encoding size}
	\end{subfigure}
	\begin{subfigure}{.7\textwidth}
		\centering
		\includegraphics[width=1\textwidth]{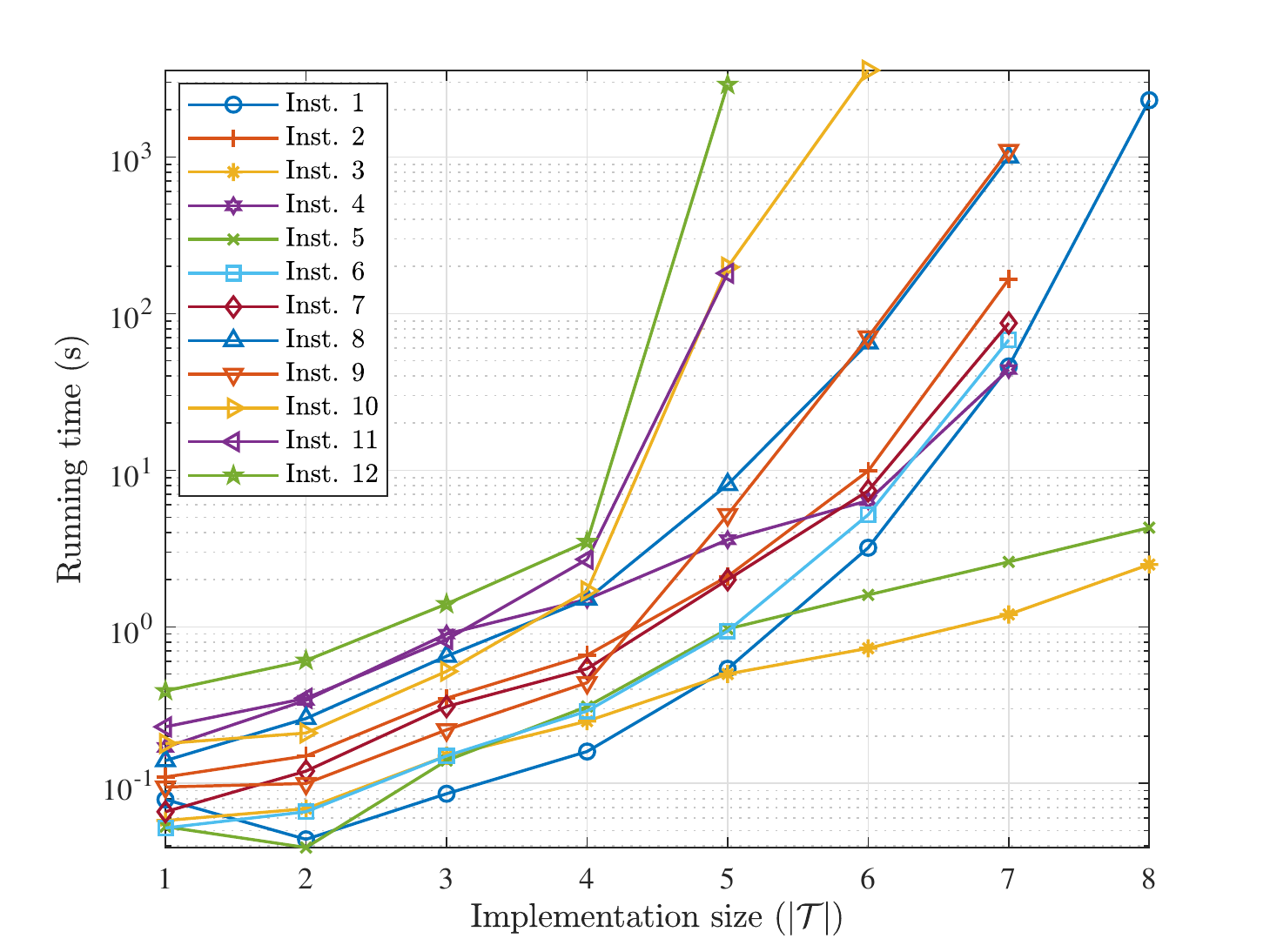}\caption{\scriptsize Running time}
	\end{subfigure}
	\begin{subfigure}{.7\textwidth}
	\centering
	\includegraphics[width=1\textwidth]{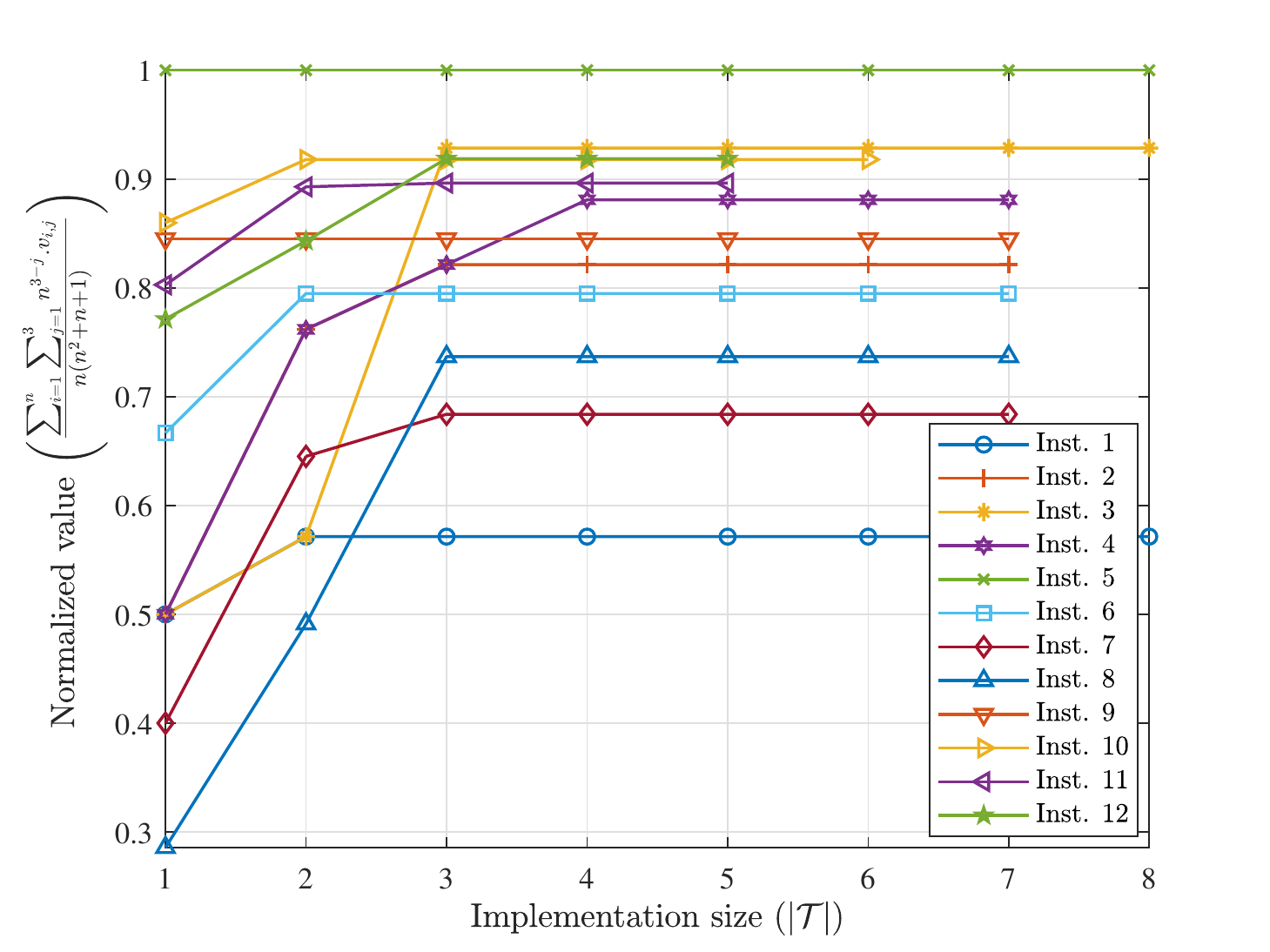}\caption{\scriptsize Normalized value}
    \end{subfigure}
	\caption{Results of applying our method to the instances in Table~\ref{tab:def-inst}, with different bounds on implementation size $|\trans|$. (a) shows the size of the MaxSAT encoding as the number of variables (solid lines) and the number of clauses (dashed lines). (b) shows the running time of the MaxSAT solver plus the time for the encoding. (c) shows the level of realizability of soft specifications.}
	\label{fig:inst-stat}
\end{figure}

In Figure~\ref{fig:inst-stat}, we show the results for the instances defined in Table~\ref{tab:def-inst} (detailed results are reported in Table~\ref{tab:inst-stat}). As expected, the value function is monotonically nondecreasing with respect to the bound on the implementation size. In the first set of instances, the specifications have large number of variables (due to full connectivity), and the bottleneck is the translation to automata.
In the third set of instances, the limiting factor is the number of soft specifications, leading to large weights and number of variables in the MaxSAT formulation.
We observe that the number of soft specifications is an important factor affecting the scalability of the proposed method. For example, Instance 12, on which the MaxSAT solver reaches time-out for implementation size bound 6 contains 23 soft specifications.

\section{Conclusion and Future Work}\label{sec:conclusion}
In this paper, we considered settings in which a system's requirements are categorized as hard and soft linear temporal logic (LTL) specifications and the goal is to design a controller that satisfies the hard specifications while maximally realizing the soft specifications. To that end, we introduced relaxations of soft LTL formulas and accordingly defined a value function that captures the level of realizing a conjunction of soft LTL formulas. We further constructed a MaxSAT encoding of maximum realizability that aims to maximize the value function. By incrementing the size of the implementation and generating the induced MaxSAT encodings, we developed a bounded synthesis procedure to find a controller with smallest size that meets a termination criterion. We computed a theoretical bound on the size of the implementation and proved soundness and completeness of the synthesis algorithm. Additionally, we discussed multiple generalizations of the proposed method and provided experimental results in multiple scenarios.

As part of future work, we plan to employ the proposed algorithm to construct controllers from a combination of temporal logic specifications and data in the form of sample demonstration of desired system behavior. In such settings, the system is asked to imitate the sample demonstrations as much as possible while satisfying the given specifications. We are also considering to design a customized search procedure for solving MaxSAT instances generated for maximum realizability that benefits from the knowledge on the specific structure of soft clauses.

\clearpage
\begin{acknowledgements}
This work was supported in part by AFRL grants UTC 17-S8401-10-C1 and FA8650-15-C-2546, and ONR grant N000141613165.
\end{acknowledgements}

\bibliographystyle{ieeetr}
\bibliography{main.bib}

\end{document}